\documentclass[fleqn]{llncs}
\usepackage{pgamd}

\title{Programming an Interpreter Using\\ Molecular Dynamics%
       \thanks{This research was partly carried out in the framework of
               the Jacquard-project Symbiosis, which is funded by the
               Netherlands Organisation for Scientific Research (NWO).}}
\author{J.A. Bergstra\inst{1}\fnmsep\inst{2}
        \and
        C.A. Middelburg\inst{1}
       }
\institute{Programming Research Group,
           University of Amsterdam, \\
           P.O.~Box~41882, 1009~DB~Amsterdam, the Netherlands \\
           \and
           Department of Philosophy,
           Utrecht University, \\
           P.O.~Box~80126, 3508~TC~Utrecht, the Netherlands \\
           \email{J.A.Bergstra@uva.nl,C.A.Middelburg@uva.nl}
          }

\begin{document}
\maketitle

\begin{abstract}
\sloppy
\PGA\ (ProGram Algebra) is an algebra of programs which concerns
programs in their simplest form: sequences of instructions.
Molecular dynamics is a simple model of computation developed in the
setting of \PGA, which bears on the use of dynamic data structures in
programming.
We consider the programming of an interpreter for a program notation
that is close to existing assembly languages using \PGA\ with the
primitives of molecular dynamics as basic instructions.
It happens that, although primarily meant for explaining programming
language features relating to the use of dynamic data structures, the
collection of primitives of molecular dynamics in itself is suited to
our programming wants.
\\[1.5ex]
{\sl Keywords:}
interpreter, program algebra, molecular dynamics,
thread algebra, action transforming service.
\\[1.5ex]
{\sl 1998 ACM Computing Classification:}
D.1.4, D.3.1, D.3.4, F.1.1, F.3.2.
\end{abstract}

\section{Introduction}
\label{sect-intro}

In this paper, we consider the programming of an interpreter for a
program notation that is close to existing assembly languages using
\PGA\ (ProGram Algebra).
With that we carry on the line of research with which a start was made
in~\cite{BL02a}.
The object pursued with that line of research is the development of a
theoretical understanding of possible forms of sequential programs,
starting from the simplest form of sequential programs, and associated
ways of programming.
The view is taken that sequential programs in the simplest form are
sequences of instructions.
\PGA, an algebra of programs in which programs are looked upon as
sequences of instructions, is taken for the basis of the development
aimed at.
Until now, the research was rather focussed on the development of a
theoretical understanding of possible forms of sequential programs.
The work presented in the current paper is primarily concerned with
programming using the simplest form of sequential programs and molecular
dynamics.

For the programming of the interpreter, we use \PGA\ with the primitives
of molecular dynamics as basic instructions.
Molecular dynamics is a simple model of computation bearing on the use
of dynamic data structures in programming.
In this model, states of computations resemble collections of molecules
composed of atoms and computations take place by means of actions which
transform the structure of molecules like in chemical reactions.
Molecular dynamics has been developed mainly in the setting of \PGA.
The model introduced in the current paper is a small expansion of the
model that was first described informally in~\cite{BB02a} and formally
in~\cite{BM06c}.

In the line of research carried on, the view is taken that the
behaviours exhibited by sequential programs on execution are threads as
considered in \BTA\ (Basic Thread Algebra).%
\footnote
{In~\cite{BL02a}, basic thread algebra is introduced under the name
 basic polarized process algebra.
 Prompted by the development of thread algebra~\cite{BM04c}, which is a
 design on top of it, basic polarized process algebra has been renamed
 to basic thread algebra.
}
A thread proceeds by doing steps in a sequential fashion.
A thread may do certain steps for having itself affected by some service
or affecting some service.
In the current paper, we will use a generalization of the use mechanism
introduced in~\cite{BM04c} and a complementary mechanism of that
mechanism for such interactions between threads that are behaviours
exhibited by sequential programs on execution and services that deal
with steps that relate to molecular dynamics.
A slightly different generalization of the use mechanism
from~\cite{BM04c} is introduced in~\cite{BM07f} under the name action
transforming thread-service composition.

A hierarchy of program notations rooted in \PGA\ is introduced
in~\cite{BL02a}.
In the current paper, we consider the programming of an interpreter for
one program notation that belongs to this hierarchy.
The program notation in question, called \PGLD\ (ProGramming
Language~D), is a very simple program notation which is close to
existing assembly languages.
The hierarchy also includes a program notation, called \PGLS\
(ProGramming Language for Structured programming), that supports
structured programming by offering a rendering of conditional and loop
constructs instead of (unstructured) jump instructions.
Each \PGLS\ program can be translated into a semantically equivalent
\PGLD\ program.
In~\cite{BM07e}, a variant of \PGLD\ with indirect jump instructions is
introduced.
In the current paper, we show how to adapt the interpreter for \PGLD\ to
the presence of indirect jump instructions.

This paper is organized as follows.
First, we review \BTA, \PGA, and \PGLD\ (Sections~\ref{sect-BTA},
\ref{sect-PGA}, and~\ref{sect-PGLD}).
Next, we extend \BTA\ with the mechanisms for interaction between
threads and services used, introduce a state-based approach to describe
services, and give a state-based description of services for molecular
dynamics (Sections~\ref{sect-TSI}, \ref{sect-service-descr},
and~\ref{sect-MDS}).
Following this, we give \PGA\ programs for creating representations of
\PGLD\ programs by molecules and a \PGA\ program for interpreting those
representations (Sections~\ref{sect-representation}
and~\ref{sect-interpretation}).
After that, we introduce the variant of \PGLD\ with indirect jump
instructions and adapt the above-mentioned \PGA\ programs to the
presence of indirect jump instructions (Sections~\ref{sect-PGLDij}
and~\ref{sect-interpretation-ij}).
Finally, we make some concluding remarks (Section~\ref{sect-concl}).

\section{Basic Thread Algebra}
\label{sect-BTA}

In this section, we review \BTA, a form of process algebra which is
tailored to the description of the behaviour of deterministic sequential
programs under execution.
The behaviours concerned are called \emph{threads}.

In \BTA, it is assumed that there is a fixed but arbitrary finite set of
\emph{basic actions} $\BAct$ with $\Tau \not\in \BAct$.
We write $\BActTau$ for $\BAct \union \set{\Tau}$.
The members of $\BActTau$ are referred to as \emph{actions}.

The intuition is that each basic action performed by a thread is taken
as a command to be processed by a service provided by the execution
environment of the thread.
The processing of a command may involve a change of state of the service
concerned.
At completion of the processing of the command, the service produces a
reply value.
This reply is either $\True$ or $\False$ and is returned to the thread
concerned.

Although \BTA\ is one-sorted, we make this sort explicit.
The reason for this is that we will extend \BTA\ with an additional sort
in Section~\ref{sect-TSI}.

The algebraic theory \BTA\ has one sort: the sort $\Thr$ of
\emph{threads}.
To build terms of sort $\Thr$, \BTA\ has the following constants and
operators:
\begin{iteml}
\item
the \emph{deadlock} constant $\const{\DeadEnd}{\Thr}$;
\item
the \emph{termination} constant $\const{\Stop}{\Thr}$;
\item
for each $a \in \BActTau$, the binary \emph{postconditional composition}
operator $\funct{\pcc{\ph}{a}{\ph}}{\linebreak[2]\Thr \x \Thr}{\Thr}$.
\end{iteml}
Terms of sort $\Thr$ are built as usual (see e.g.~\cite{ST99a,Wir90a}).
Throughout the paper, we assume that there are infinitely many variables
of sort $\Thr$, including $x,y,z$.

We use infix notation for postconditional composition.
We introduce \emph{action prefixing} as an abbreviation: $a \bapf p$,
where $p$ is a term of sort $\Thr$, abbreviates $\pcc{p}{a}{p}$.

Let $p$ and $q$ be closed terms of sort $\Thr$ and $a \in \BActTau$.
Then $\pcc{p}{a}{q}$ will perform action $a$, and after that proceed as
$p$ if the processing of $a$ leads to the reply $\True$ (called a
positive reply), and proceed as $q$ if the processing of $a$ leads to
the reply $\False$ (called a negative reply).
The action $\Tau$ plays a special role.
It is a concrete internal action: performing $\Tau$ will never lead to a
state change and always lead to a positive reply, but notwithstanding
all that its presence matters.

\BTA\ has only one axiom.
This axiom is given in Table~\ref{axioms-BTA}.%
\begin{table}[!tb]
\caption{Axiom of \BTA}
\label{axioms-BTA}
\begin{eqntbl}
\begin{axcol}
\pcc{x}{\Tau}{y} = \pcc{x}{\Tau}{x}                      & \axiom{T1}
\end{axcol}
\end{eqntbl}
\end{table}
Using the abbreviation introduced above, axiom T1 can be written as
follows: $\pcc{x}{\Tau}{y} = \Tau \bapf x$.

Each closed \BTA\ term of sort $\Thr$ denotes a finite thread, i.e.\ a
thread of which the length of the sequences of actions that it can
perform is bounded.
Guarded recursive specifications give rise to infinite threads.

A \emph{guarded recursive specification} over \BTA\ is a set of
recursion equations $E = \set{X = t_X \where X \in V}$, where $V$ is a
set of variables of sort $\Thr$ and each $t_X$ is a term of the form
$\DeadEnd$, $\Stop$ or $\pcc{t}{a}{t'}$ with $t$ and $t'$ \BTA\ terms of
sort $\Thr$ that contain only variables from $V$.
We write $\vars(E)$ for the set of all variables that occur on the
left-hand side of an equation in $E$.
We are only interested in models of \BTA\ in which guarded recursive
specifications have unique solutions, such as the projective limit model
of \BTA\ presented in~\cite{BB03a}.
A thread that is the solution of a finite guarded recursive
specification over \BTA\ is called a \emph{finite-state} thread.

We extend \BTA\ with guarded recursion by adding constants for solutions
of guarded recursive specifications and axioms concerning these
additional constants.
For each guarded recursive specification $E$ and each $X \in \vars(E)$,
we add a constant of sort $\Thr$ standing for the unique solution of $E$
for $X$ to the constants of \BTA.
The constant standing for the unique solution of $E$ for $X$ is denoted
by $\rec{X}{E}$.
Moreover, we add the axioms for guarded recursion given in
Table~\ref{axioms-REC} to \BTA,%
\begin{table}[!tb]
\caption{Axioms for guarded recursion}
\label{axioms-REC}
\begin{eqntbl}
\begin{saxcol}
\rec{X}{E} = \rec{t_X}{E} & \mif X \!=\! t_X \in E       & \axiom{RDP}
\\
E \Implies X = \rec{X}{E} & \mif X \in \vars(E)          & \axiom{RSP}
\end{saxcol}
\end{eqntbl}
\end{table}
where we write $\rec{t_X}{E}$ for $t_X$ with, for all $Y \in \vars(E)$,
all occurrences of $Y$ in $t_X$ replaced by $\rec{Y}{E}$.
In this table, $X$, $t_X$ and $E$ stand for an arbitrary variable of
sort $\Thr$, an arbitrary \BTA\ term of sort $\Thr$ and an arbitrary
guarded recursive specification over \BTA, respectively.
Side conditions are added to restrict the variables, terms and guarded
recursive specifications for which $X$, $t_X$ and $E$ stand.
The equations $\rec{X}{E} = \rec{t_X}{E}$ for a fixed $E$ express that
the constants $\rec{X}{E}$ make up a solution of $E$.
The conditional equations $E \Implies X = \rec{X}{E}$ express that this
solution is the only one.

We will use the following abbreviation: $a^\omega$, where
$a \in \BActTau$, abbreviates $\rec{X}{\set{X = a \bapf X}}$.

We will write \BTA+\REC\ for \BTA\ extended with the constants for
solutions of guarded recursive specifications and axioms RDP and RSP.

In~\cite{BM05c}, we show that the threads considered in \BTA+\REC\ can
be viewed as processes that are definable over ACP~\cite{Fok00}.

Closed terms of sort $\Thr$ from the language of \BTA+\REC\ that denote
the same infinite thread cannot always be proved equal by means of the
axioms of \BTA+\REC.
We introduce the approximation induction principle to remedy this.
The approximation induction principle, \AIP\ in short, is based on the
view that two threads are identical if their approximations up to any
finite depth are identical.
The approximation up to depth $n$ of a thread is obtained by cutting it
off after performing a sequence of actions of length $n$.

\AIP\ is the infinitary conditional equation given in
Table~\ref{axioms-AIP}.%
\begin{table}[!tb]
\caption{Approximation induction principle}
\label{axioms-AIP}
\begin{eqntbl}
\begin{axcol}
\AND{n \geq 0} \proj{n}{x} = \proj{n}{y} \Implies x = y   & \axiom{AIP}
\end{axcol}
\end{eqntbl}
\end{table}
Here, following~\cite{BL02a}, approximation of depth $n$ is phrased in
terms of a unary \emph{projection} operator
$\funct{\pi_n}{\Thr}{\Thr}$.
The axioms for the projection operators are given in
Table~\ref{axioms-proj}.%
\begin{table}[!tb]
\caption{Axioms for projection operators}
\label{axioms-proj}
\begin{eqntbl}
\begin{axcol}
\proj{0}{x} = \DeadEnd                                   & \axiom{P0} \\
\proj{n+1}{\Stop} = \Stop                                & \axiom{P1} \\
\proj{n+1}{\DeadEnd} = \DeadEnd                          & \axiom{P2} \\
\proj{n+1}{\pcc{x}{a}{y}} =
                       \pcc{\proj{n}{x}}{a}{\proj{n}{y}} & \axiom{P3}
\end{axcol}
\end{eqntbl}
\end{table}
In this table, $a$ stands for an arbitrary member of $\BActTau$.

We will write \BTA+\REC+\AIP\ for \BTA+\REC\ extended with the
projection operators and the axioms from Tables~\ref{axioms-AIP}
and~\ref{axioms-proj}.

\section{Program Algebra}
\label{sect-PGA}

In this section, we review \PGA, an algebra of sequential programs based
on the idea that sequential programs are in essence sequences of
instructions.
\PGA\ provides a program notation for finite-state threads.

In \PGA, it is assumed that there is a fixed but arbitrary finite set
$\BInstr$ of \emph{basic instructions}.
\PGA\ has the following \emph{primitive instructions}:
\begin{iteml}
\item
for each $a \in \BInstr$, a \emph{plain basic instruction} $a$;
\item
for each $a \in \BInstr$, a \emph{positive test instruction} $\ptst{a}$;
\item
for each $a \in \BInstr$, a \emph{negative test instruction} $\ntst{a}$;
\item
for each $l \in \Nat$, a \emph{forward jump instruction} $\fjmp{l}$;
\item
a \emph{termination instruction} $\halt$.
\end{iteml}
We write $\PInstr$ for the set of all primitive instructions.

The intuition is that the execution of a basic instruction $a$ may
modify a state and produces $\True$ or $\False$ at its completion.
In the case of a positive test instruction $\ptst{a}$, basic instruction
$a$ is executed and execution proceeds with the next primitive
instruction if $\True$ is produced and otherwise the next primitive
instruction is skipped and execution proceeds with the primitive
instruction following the skipped one.
In the case where $\True$ is produced and there is not at least one
subsequent primitive instruction and in the case where $\False$ is
produced and there are not at least two subsequent primitive
instructions, deadlock occurs.
In the case of a negative test instruction $\ntst{a}$, the role of the
value produced is reversed.
In the case of a plain basic instruction $a$, the value produced is
disregarded: execution always proceeds as if $\True$ is produced.
The effect of a forward jump instruction $\fjmp{l}$ is that execution
proceeds with the $l$-th next instruction of the program concerned.
If $l$ equals $0$ or the $l$-th next instruction does not exist, then
$\fjmp{l}$ results in deadlock.
The effect of the termination instruction $\halt$ is that execution
terminates.

\PGA\ has the following constants and operators:
\begin{iteml}
\item
for each $u \in \PInstr$, an \emph{instruction} constant $u$\,;
\item
the binary \emph{concatenation} operator $\ph \conc \ph$\,;
\item
the unary \emph{repetition} operator $\ph\rep$\,.
\end{iteml}
Terms are built as usual.
Throughout the paper, we assume that there are infinitely many
variables, including $x,y,z$.

We use infix notation for concatenation and postfix notation for
repetition.

Closed \PGA\ terms are considered to denote programs.
The intuition is that a program is in essence a non-empty, finite or
periodic infinite sequence of primitive instructions.%
\footnote
{A periodic infinite sequence is an infinite sequence with only finitely
 many subsequences.}
These sequences are called \emph{single pass instruction sequences}
because \PGA\ has been designed to enable single pass execution of
instruction sequences: each instruction can be dropped after it has been
executed.
Programs are considered to be equal if they represent the same single
pass instruction sequence.
The axioms for instruction sequence equivalence are given in
Table~\ref{axioms-PGA}.%
\begin{table}[!t]
\caption{Axioms of \PGA}
\label{axioms-PGA}
\begin{eqntbl}
\begin{axcol}
(x \conc y) \conc z = x \conc (y \conc z)              & \axiom{PGA1} \\
(x^n)\rep = x\rep                                      & \axiom{PGA2} \\
x\rep \conc y = x\rep                                  & \axiom{PGA3} \\
(x \conc y)\rep = x \conc (y \conc x)\rep              & \axiom{PGA4}
\end{axcol}
\end{eqntbl}
\end{table}
In this table, $n$ stands for an arbitrary natural number greater than
$0$.
For each $n > 0$, the term $x^n$ is defined by induction on $n$ as
follows: $x^1 = x$ and $x^{n+1} = x \conc x^n$.
The \emph{unfolding} equation $x\rep = x \conc x\rep$ is
derivable.
Each closed \PGA\ term is derivably equal to a term in
\emph{canonical form}, i.e.\ a term of the form $P$ or $P \conc Q\rep$,
where $P$ and $Q$ are closed \PGA\ terms that do not contain the
repetition operator.

Each closed \PGA\ term is considered to denote a program of which the
behaviour is a finite-state thread, taking the set $\BInstr$ of basic
instructions for the set $\BAct$ of actions.
The \emph{thread extraction} operator $\extr{\ph}$ assigns a thread to
each program.
The thread extraction operator is defined by the equations given in
Table~\ref{axioms-thread-extr} (for $a \in \BInstr$, $l \in \Nat$ and
$u \in \PInstr$)%
\begin{table}[!t]
\caption{Defining equations for thread extraction operator}
\label{axioms-thread-extr}
\begin{eqntbl}
\begin{eqncol}
\extr{a} = a \bapf \DeadEnd \\
\extr{a \conc x} = a \bapf \extr{x} \\
\extr{\ptst{a}} = a \bapf \DeadEnd \\
\extr{\ptst{a} \conc x} =
\pcc{\extr{x}}{a}{\extr{\fjmp{2} \conc x}} \\
\extr{\ntst{a}} = a \bapf \DeadEnd \\
\extr{\ntst{a} \conc x} =
\pcc{\extr{\fjmp{2} \conc x}}{a}{\extr{x}}
\end{eqncol}
\qquad
\begin{eqncol}
\extr{\fjmp{l}} = \DeadEnd \\
\extr{\fjmp{0} \conc x} = \DeadEnd \\
\extr{\fjmp{1} \conc x} = \extr{x} \\
\extr{\fjmp{l+2} \conc u} = \DeadEnd \\
\extr{\fjmp{l+2} \conc u \conc x} = \extr{\fjmp{l+1} \conc x} \\
\extr{\halt} = \Stop \\
\extr{\halt \conc x} = \Stop
\end{eqncol}
\end{eqntbl}
\end{table}
and the rule given in Table~\ref{rule-thread-extr}.%
\begin{table}[!t]
\caption{Rule for cyclic jump chains}
\label{rule-thread-extr}
\begin{eqntbl}
\begin{eqncol}
x \scongr \fjmp{0} \conc y \Implies \extr{x} = \DeadEnd
\end{eqncol}
\end{eqntbl}
\end{table}
This rule is expressed in terms of the \emph{structural congruence}
predicate $\ph \scongr \ph$, which is defined by the formulas given in
Table~\ref{axioms-scongr} (for $n,m,l \in \Nat$ and
$u_1,\ldots,u_n,v_1,\ldots,v_{m+1} \in \PInstr$).%
\begin{table}[!t]
\caption{Defining formulas for structural congruence predicate}
\label{axioms-scongr}
\begin{eqntbl}
\begin{eqncol}
\fjmp{n+1} \conc u_1 \conc \ldots \conc u_n \conc \fjmp{0}
\scongr
\fjmp{0} \conc u_1 \conc \ldots \conc u_n \conc \fjmp{0}
\\
\fjmp{n+1} \conc u_1 \conc \ldots \conc u_n \conc \fjmp{m}
\scongr
\fjmp{m+n+1} \conc u_1 \conc \ldots \conc u_n \conc \fjmp{m}
\\
(\fjmp{n+l+1} \conc u_1 \conc \ldots \conc u_n)\rep \scongr
(\fjmp{l} \conc u_1 \conc \ldots \conc u_n)\rep
\\
\fjmp{m+n+l+2} \conc u_1 \conc \ldots \conc u_n \conc
(v_1 \conc \ldots \conc v_{m+1})\rep \scongr {} \\ \hfill
\fjmp{n+l+1} \conc u_1 \conc \ldots \conc u_n \conc
(v_1 \conc \ldots \conc v_{m+1})\rep
\\
x \scongr x
\\
x_1 \scongr y_1 \And x_2 \scongr y_2 \Implies
x_1 \conc x_2 \scongr y_1 \conc y_2 \And
{x_1}\rep \scongr {y_1}\rep
\end{eqncol}
\end{eqntbl}
\end{table}

The equations given in Table~\ref{axioms-thread-extr} do not cover the
case where there is a cyclic chain of forward jumps.
Programs are structural congruent if they are the same after removing
all chains of forward jumps in favour of single jumps.
Because a cyclic chain of forward jumps corresponds to $\fjmp{0}$,
the rule from Table~\ref{rule-thread-extr} can be read as follows:
if $x$ starts with a cyclic chain of forward jumps, then $\extr{x}$
equals $\DeadEnd$.
It is easy to see that the thread extraction operator assigns the same
thread to structurally congruent programs.
Therefore, the rule from Table~\ref{rule-thread-extr} can be replaced by
the following generalization:
$x \scongr y  \Implies \extr{x} = \extr{y}$.

Let $E$ be a finite guarded recursive specification over \BTA, and let
$P_X$ be a closed \PGA\ term for each $X \in \vars(E)$.
Let $E'$ be the set of equations that results from replacing in $E$ all
occurrences of $X$ by $\extr{P_X}$ for each $X \in \vars(E)$.
If $E'$ can be obtained by applications of axioms PGA1--PGA4, the
defining equations for the thread extraction operator and the rule for
cyclic jump chains, then $\extr{P_X}$ is the solution of $E$ for $X$.
Such a finite guarded recursive specification can always be found.
%
%
Thus, the behaviour of each closed \PGA\ term, is a thread that is
definable by a finite guarded recursive specification over \BTA.
Moreover, each finite guarded recursive specification over \BTA\ can be
translated to a closed \PGA\ term of which the behaviour is the solution
of the finite guarded recursive specification concerned
(cf.\ Section~4 of~\cite{PZ06a}).

Closed \PGA\ terms are loosely called \PGA\ \emph{programs}.
\PGA\ programs in which the repetition operator do not occur
are called \emph{finite} \PGA\ programs.
%

\section{The Program Notation \PGLD}
\label{sect-PGLD}

In this section, we review a program notation which is rooted in \PGA.
This program notation, called \PGLD, belongs to a hierarchy of program
notations introduced in~\cite{BL02a}.
\PGLD\ is close to existing assembly languages.
It has absolute jump instructions and no explicit termination
instruction.

In \PGLD, like in \PGA, it is assumed that there is a fixed but
arbitrary finite set of \emph{basic instructions} $\BInstr$.
Again, the intuition is that the execution of a basic instruction $a$
may modify a state and produces $\True$ or $\False$ at its completion.

\PGLD\ has the following primitive instructions:
\begin{itemize}
\item
for each $a \in \BInstr$, a \emph{plain basic instruction} $a$;
\item
for each $a \in \BInstr$, a \emph{positive test instruction} $\ptst{a}$;
\item
for each $a \in \BInstr$, a \emph{negative test instruction} $\ntst{a}$;
\item
for each $l \in \Nat$, a \emph{direct absolute jump instruction}
$\ajmp{l}$.
\end{itemize}
\PGLD\ programs have the form $u_1;\ldots;u_k$, where $u_1,\ldots,u_k$
are primitive instructions of \PGLD.

The plain basic instructions, the positive test instructions, and the
negative test instructions are as in \PGA.
The effect of a direct absolute jump instruction $\ajmp{l}$ is that
execution proceeds with the $l$-th instruction of the program concerned.
If $\ajmp{l}$ is itself the $l$-th instruction, then deadlock occurs.
If $l$ equals $0$ or $l$ is greater than the length of the program, then
termination occurs.

We define the meaning of \PGLD\ programs by means of a function
$\pgldpga$ from the set of all \PGLD\ programs to the set of all \PGA\
programs.
This function is defined by
\begin{ldispl}
\pgldpga(u_1 \conc \ldots \conc u_k) =
(\psi_1(u_1) \conc \ldots \conc \psi_k(u_k) \conc
 \halt \conc \halt)\rep\;,
\end{ldispl}
where the auxiliary functions $\psi_j$ from the set of all primitive
instructions of \PGLD\ to the set of all primitive instructions of \PGA\
are defined as follows ($1 \leq j \leq k$):
\begin{ldispl}
\begin{aceqns}
\psi_j(\ajmp{l}) & = & \fjmp{l-j}       & \mif j \leq l \leq k\;, \\
\psi_j(\ajmp{l}) & = & \fjmp{k+2-(j-l)} & \mif 0   <  l   <  j\;, \\
\psi_j(\ajmp{l}) & = & \halt            & \mif l = 0 \Or l > k\;, \\
\psi_j(u)        & = & u
                    & \mif u\; \mathrm{is\;not\;a\;jump\;instruction}\;.
\end{aceqns}
\end{ldispl}

\sloppy
Let $P$ be a \PGLD\ program.
Then $\pgldpga(P)$ represents the meaning of $P$ as a \PGA\ program.
The intended behaviour of $P$ is the behaviour of $\pgldpga(P)$.
That is, the \emph{behaviour} of $P$, written $\extr{P}_\sPGLD$, is
$\extr{\pgldpga(P)}$.

We use the phrase \emph{projection semantics} to refer to the approach
to semantics followed in this section.
The meaning function $\pgldpga$ is called a \emph{projection}.

\section{Threads Interacting with Action Transforming Services}
\label{sect-TSI}

A thread may perform certain actions for having itself affected by some
service or affecting some service.
When processing an action performed by a thread, an action transforming
service affects that thread in one of the following ways:
(i)~by returning a reply value to the thread at completion of the
processing of the action performed by the thread;
(ii)~by turning the processed action into another action.
In this section, we introduce a mechanism that allows for services to
affect threads in either way.
The mechanism is a generalization of the use mechanism introduced
in~\cite{BM04c}.%
\footnote
{In later papers, the original use mechanism is also called
 thread-service composition.}
We also introduce a complementary mechanism of that generalized use
mechanism and an abstraction mechanism.
The difference between the generalized use mechanism and the
complementary mechanism is a matter of perspective:
the former is concerned with the effect of services on threads and
therefore produces threads, whereas the latter is concerned with the
effect of threads on services and therefore produces services.
The abstraction mechanism serves for concealment of the presence of
internal actions, which arise primarily from the generalized use
mechanism.

We will use the generalized use mechanism and the complementary
mechanism of that mechanism for interactions between threads that are
behaviours exhibited by sequential programs on execution and services
that process actions that relate to molecular dynamics.

It is assumed that there is a fixed but arbitrary finite set of
\emph{foci} $\Foci$ and a fixed but arbitrary finite set of
\emph{methods} $\Meth$.
Each focus plays the role of a name of a service provided by the
execution environment that can be requested to process a command.
Each method plays the role of a command proper.
For the set $\BAct$ of actions, we take the set
$\set{f.m \where f \in \Foci, m \in \Meth}$.
Performing an action $f.m$ is taken as making a request to the
service named $f$ to process command $m$.

We introduce yet another sort: the sort $\Serv$ of \emph{services}.
$\Serv$ is considered to stand for the set of all services.
We identify services with pairs $\tup{H_1,H_2}$, where
$\funct{H_1}{\neseqof{\Meth}}{\set{\True,\False,\Mless,\Blocked}}$ and
$\funct{H_2}{\neseqof{\Meth}}{\BActTau}$,
satisfying the following conditions:
\begin{ldispl}
\Forall{m \in \Meth}{{}}
\\ \quad
 {(\Exists{\alpha \in \seqof{\Meth}}
    {H_1(\alpha \concat \seq{m}) = \Mless} \Implies
   \Forall{\alpha' \in \seqof{\Meth}}
    {H_1(\alpha' \concat \seq{m}) \not\in \set{\True,\False}})}\;,
\eqnsep
\Forall{\alpha \in \neseqof{\Meth},m \in \Meth}
 {(H_1(\alpha) = \Blocked \Implies H_1(\alpha \concat \seq{m}) =
   \Blocked)}\;,
\eqnsep
\Forall{\alpha \in \neseqof{\Meth}}
 {(H_1(\alpha) \neq \Mless \Iff H_2(\alpha) = \Tau)}\;.
\end{ldispl}

Let $H$ be a service, and let $H_1$ and $H_2$ be the unique functions
such that $H = \tup{H_1,H_2}$.
Then we write $\rfunc{H}$ and $\afunc{H}$ for $H_1$ and $H_2$,
respectively.
Given a service $H$ and a method $m \in \Meth$,
the \emph{derived service} of $H$ after processing $m$,
written $\derive{m}H$, is defined by
$\rfunc{\derive{m}H}(\alpha) = \rfunc{H}(\seq{m} \concat \alpha)$ and
$\afunc{\derive{m}H}(\alpha) = \afunc{H}(\seq{m} \concat \alpha)$.

A service $H$ can be understood as follows:
\begin{iteml}
\item
if $\rfunc{H}(\seq{m}) = \True$, then the request to process $m$ is
accepted by the service, a positive reply is produced, $m$ is turned into
$\Tau$, and the service proceeds as $\derive{m}H$;
\item
if $\rfunc{H}(\seq{m}) = \False$, then the request to process $m$ is
accepted by the service, a negative reply is produced, $m$ is turned into
$\Tau$, and the service proceeds as $\derive{m}H$;
\item
if $\rfunc{H}(\seq{m}) = \Mless$, then the request to process $m$ is
accepted by the service, no reply is produced, $m$ is turned into
$\afunc{H}(\seq{m})$, and the service proceeds as $\derive{m}H$;
\item
if $\rfunc{H}(\seq{m}) = \Blocked$, then the request to process $m$ is
rejected by the service.
\end{iteml}
The three conditions imposed on services can be paraphrased as follows:
\begin{iteml}
\item
if it is possible that no reply is produced at completion of the
processing of a command, then it is impossible that a positive or
negative reply is produced at completion of the processing of that
command;
\item
after a request to process a command has been rejected, any request to
process a command will be rejected;
\item
a reply is produced at completion of the processing of a command if and
only if the command is turned into $\Tau$.
\end{iteml}

We introduce the following additional constant and operators:
\begin{iteml}
\item
the \emph{divergent service} constant $\const{\DivServ}{\Serv}$;
\item
for each $f \in \Foci$, the binary \emph{use} operator
$\funct{\use{\ph}{f}{\ph}}{\Thr \x \Serv}{\Thr}$;
\item
for each $f \in \Foci$, the binary \emph{apply} operator
$\funct{\apply{\ph}{f}{\ph}}{\Thr \x \Serv}{\Serv}$.
\end{iteml}
We use infix notation for use and apply.

$\DivServ$ is the unique service $H$ with the property that
$\rfunc{H}(\alpha) = \Blocked$ for all $\alpha \in \neseqof{\Meth}$.
The operators $\use{\ph}{f}{\ph}$ and $\apply{\ph}{f}{\ph}$ are
complementary.
Intuitively, $\use{p}{f}{H}$ is the thread that results from processing
all actions performed by thread $p$ that are of the form $f.m$ by the
service $H$.
When an action of the form $f.m$ performed by thread $p$ is processed by
the service $H$, that action is turned into another action and, if this
action is $\Tau$, postconditional composition is removed in favour of
action prefixing on the basis of the reply value produced.
Intuitively, $\apply{p}{f}{H}$ is the service that results from
processing all basic actions performed by thread $P$ that are of the
form $f.m$ by service $H$.
When an action of the form $f.m$ performed by thread $p$ is processed by
service $H$, that service is changed into $\derive{m}H$.

The axioms for the use and apply operators are given in
Tables~\ref{axioms-use} and~\ref{axioms-apply}.%
\begin{table}[!t]
\caption{Axioms for use operators}
\label{axioms-use}
\begin{eqntbl}
\begin{saxcol}
\use{\Stop}{f}{H} = \Stop                            & & \axiom{TSU1} \\
\use{\DeadEnd}{f}{H} = \DeadEnd                      & & \axiom{TSU2} \\
\use{\Tau \bapf x}{f}{H} =
                          \Tau \bapf (\use{x}{f}{H}) & & \axiom{TSU3} \\
\use{(\pcc{x}{g.m}{y})}{f}{H} =
\pcc{(\use{x}{f}{H})}{g.m}{(\use{y}{f}{H})}
 & \mif f \neq g                                       & \axiom{TSU4} \\
\use{(\pcc{x}{f.m}{y})}{f}{H} =
\Tau \bapf (\use{x}{f}{\derive{m}H})
                  & \mif \rfunc{H}(\seq{m}) = \True    & \axiom{TSU5} \\
\use{(\pcc{x}{f.m}{y})}{f}{H} =
\Tau \bapf (\use{y}{f}{\derive{m}H})
                  & \mif \rfunc{H}(\seq{m}) = \False   & \axiom{TSU6} \\
\use{(\pcc{x}{f.m}{y})}{f}{H} =
\use{(\pcc{x}{\afunc{H}(\seq{m})}{y})}{f}{\derive{m}H}
                  & \mif \rfunc{H}(\seq{m}) = \Mless   & \axiom{TSU7} \\
\use{(\pcc{x}{f.m}{y})}{f}{H} = \DeadEnd
                  & \mif \rfunc{H}(\seq{m}) = \Blocked & \axiom{TSU8} \\
\use{(\pcc{x}{f.m}{y})}{f}{\DivServ} = \DeadEnd      & & \axiom{TSU9}
\end{saxcol}
\end{eqntbl}
\end{table}
\begin{table}[!t]
\caption{Axioms for apply operators}
\label{axioms-apply}
\begin{eqntbl}
\begin{saxcol}
\apply{\Stop}{f}{H} = H                              & & \axiom{TSA1} \\
\apply{\DeadEnd}{f}{H} = \DivServ                    & & \axiom{TSA2} \\
\apply{(\Tau \bapf x)}{f}{H} = \apply{x}{f}{H}       & & \axiom{TSA3} \\
\apply{(\pcc{x}{g.m}{y})}{f}{H} = \DivServ
                                       & \mif f \neq g & \axiom{TSA4} \\
\apply{(\pcc{x}{f.m}{y})}{f}{H} = \apply{x}{f}{\derive{m}H}
                  & \mif \rfunc{H}(\seq{m}) = \True    & \axiom{TSA5} \\
\apply{(\pcc{x}{f.m}{y})}{f}{H} = \apply{y}{f}{\derive{m}H}
                  & \mif \rfunc{H}(\seq{m}) = \False   & \axiom{TSA6} \\
\apply{(\pcc{x}{f.m}{y})}{f}{H} =
\apply{(\pcc{x}{\afunc{H}(\seq{m})}{y})}{f}{\derive{m}H}
                  & \mif \rfunc{H}(\seq{m}) = \Mless   & \axiom{TSA7} \\
\apply{(\pcc{x}{f.m}{y})}{f}{H} = \DivServ
                  & \mif \rfunc{H}(\seq{m}) = \Blocked & \axiom{TSA8} \\
\apply{(\pcc{x}{f.m}{y})}{f}{\DivServ} = \DivServ    & & \axiom{TSA9} \\
(\AND{n \geq 0} \apply{\proj{n}{x}}{f}{H} = \DivServ) \Implies
\apply{x}{f}{H} = \DivServ                           & & \axiom{TSA10}
\end{saxcol}
\end{eqntbl}
\end{table}
In these tables, $f$ and $g$ stand for an arbitrary foci from $\Foci$,
$m$ stands for an arbitrary method from $\Meth$, and $H$ is a variable
of sort $\Serv$.
Axioms TSU3 and TSU4 express that the action $\Tau$ and actions of
the form $g.m$, where $f \neq g$, are not processed.
Axioms TSU5--TSU7 express that a thread is affected by a service
as described above when an action of the form $f.m$ performed by the
thread is processed by the service.
Axiom TSU8 expresses that deadlock takes place when an action to be
processed is not accepted.
Axiom TSU9 expresses that the divergent service does not accept any
action.
Axiom TSA3 expresses that a service is not affected by a thread when the
action $\Tau$ is performed by the thread and axiom TSA4 expresses that a
service is turned into the divergent service when an action of the form
$g.m$, where $f \neq g$,  is performed by the thread.
Axioms TSA5--TSA7 express that a service is affected by a thread as
described above when an action of the form $f.m$ performed by the thread
is processed by the service.
Axiom TSA8 expresses that a service is turned into the divergent service
when an action performed by the thread is not accepted.
Axiom TSA9 expresses that the divergent service is not affected by a
thread when an action of the form $f.m$ is performed by the thread.

We say that $\apply{p}{f}{H}$ is a \emph{divergent thread application}
if $\apply{\proj{n}{p}}{f}{H} = \DivServ$ for all $n \in \Nat$.
Axiom TSA10 can be read as follows: if $\apply{p}{f}{H}$ is a divergent
thread application, then it equals $\DivServ$.

The use operators introduced in~\cite{BM04c} deals in essence with
services $H$ where $\afunc{H}(\alpha) = \Tau$ for all
$\alpha \in \neseqof{\Meth}$.
For these services, the use operators introduced here coincide with
those use operators.

Let $T$ stand for either \BTA, \BTA+\REC\ or \BTA+\REC+\AIP.
Then we will write $T+\TSI$ for $T$, taking the set
$\set{f.m \where f \in \Foci, m \in \Meth}$ for $\BAct$, extended with
the divergent service constant, the use and apply operators, and the
axioms from Tables~\ref{axioms-use} and~\ref{axioms-apply}, with the
exception of axiom TSA10 in the case where $T$ does not stand for
\BTA+\REC+\AIP.

The action $\Tau$ is an internal action whose presence matters.
To conceal its presence in the case where it does not matter after all,
we also introduce the unary \emph{abstraction} operator
$\funct{\abstr}{\Thr}{\Thr}$.

The axioms for the abstraction operator are given in
Table~\ref{axioms-abstr}.%
\begin{table}[!t]
\caption{Axiom for abstraction}
\label{axioms-abstr}
\begin{eqntbl}
\begin{axcol}
\abstr(\Stop) = \Stop                                    & \axiom{TT1} \\
\abstr(\DeadEnd) = \DeadEnd                              & \axiom{TT2} \\
\abstr(\pcc{x}{\Tau}{y}) = \abstr(x)                     & \axiom{TT3} \\
\abstr(\pcc{x}{a}{y}) = \pcc{\abstr(x)}{a}{\abstr(y)}    & \axiom{TT4}
\end{axcol}
\end{eqntbl}
\end{table}
In this table, $a$ stands for an arbitrary basic action from $\BAct$.

Abstraction can for instance be appropriate in the case where $\Tau$
arises from turning actions of an auxiliary nature into $\Tau$ on
use of a service.
An example of this case will occur in Section~\ref{sect-interpretation}.
Unlike the use mechanisms introduced in~\cite{BM04c} and in the current
paper, the use mechanism introduced in~\cite{BP02a} incorporates
abstraction.

Let $T$ stand for either \BTA, \BTA+\REC, \BTA+\REC+\AIP, \BTA+\TSI,
BTA+\REC+\TSI\ or \BTA+\REC+\AIP+\TSI.
Then we will write $T$+\ABSTR\ for $T$ extended with the abstraction
operator and the axioms from Table~\ref{axioms-abstr}.

The equation $\abstr(\Tau^\omega) = \DeadEnd$ is derivable from the
axioms of \BTA+\REC+\linebreak[2]\AIP+\ABSTR.

To simplify matters, from now on the set
$\set{f.m \where f \in \Foci, m \in \Meth}$
is taken as the set $\BInstr$ of basic instructions when \PGA\ or \PGLD\
is concerned.
Thereby no real restriction is imposed on the set $\BInstr$: in the case
where the cardinality of $\Foci$ equals~$1$, all basic instructions have
the same focus and the set $\Meth$ of methods can be looked upon as the
set $\BInstr$ of basic instructions.
We use the convention to omit foci from \PGA\ programs in which all
basic instructions have the same focus.

Strictly speaking, the propositions and theorems presented in this paper
are proved in the algebraic theory obtained by:
(i)~combining \PGA\ with \BTA+\REC+\AIP+\TSI+\ABSTR, resulting in a
theory with three sorts: a sort $\Prog$ of programs, a sort $\Thr$ of
threads, and a sort $\Serv$ of services;
(ii)~extending the result by taking $\extr{\ph}$ for an additional
operator from sort $\Prog$ to sort $\Thr$ and taking the semantic
equations and rule defining thread extraction for
additional\linebreak[2] axioms.

\section{State-Based Description of Services}
\label{sect-service-descr}

In this section, we introduce the state-based approach to describe
families of services that will be used in Section~\ref{sect-MDS}.
This approach is similar to the approach to describe state machines
introduced in~\cite{BP02a}.

In this approach, a family of services is described by
\begin{itemize}
\item
a set of states $S$;
\item
an effect function $\funct{\eff}{\Meth \x S}{S}$;
\item
a yield function
$\funct{\yld}{\Meth \x S}{\set{\True,\False,\Mless,\Blocked}}$;
\item
an action function
$\funct{\act}{\Meth \x S}{\BActTau}$;
\end{itemize}
satisfying the following conditions:
\pagebreak[2]
\begin{ldispl}
\Forall{m \in \Meth}
 {(\Exists{s \in S}{\yld(m,s) = \Mless} \Implies
   \Forall{s' \in S}{\yld(m,s') \not\in \set{\True,\False}})}\;,
\eqnsep
\Exists{s \in S}
 {\Forall{m \in \Meth}
   {{} \\ \quad
    (\yld(m,s) = \Blocked \And
     \Forall{s' \in S}
      {(\yld(m,s') = \Blocked \Implies \eff(m,s') = s)})}}\;,
\eqnsep
\Forall{m \in \Meth,s \in S}
 {(\yld(m,s) \neq \Mless \Iff \act(m,s) = \Tau)}\;.
\end{ldispl}
The set $S$ contains the states in which the services may be, and the
functions $\eff$, $\yld$ and $\act$ give, for each method $m$ and state
$s$, the state, reply and action, respectively, that result from
processing $m$ in state $s$.

We define, for each $s \in S$, a cumulative effect function
$\funct{\ceff_s}{\seqof{\Meth}}{S}$ in terms of $s$ and $\eff$ as follows:
\begin{ldispl}
\ceff_s(\emptyseq) = s\;,
\\
\ceff_s(\alpha \concat \seq{m}) = \eff(m,\ceff_s(\alpha))\;.
\end{ldispl}
We define, for each $s \in S$, a service $H_s$ in terms of $\ceff_s$,
$\yld$ and $\act$ as follows:
\begin{ldispl}
\rfunc{H_s}(\alpha \concat \seq{m})  = \yld(m,\ceff_s(\alpha))\;,
\\
\afunc{H_s}(\alpha \concat \seq{m}) = \act(m,\ceff_s(\alpha))\;.
\end{ldispl}
$H_s$ is called the service with \emph{initial state} $s$ described by
$S$, $\eff$, $\yld$ and $\act$.
We say that $\set{H_s \where s \in S}$ is the \emph{family of services}
described by $S$, $\eff$, $\yld$ and $\act$.

The conditions that are imposed on $S$, $\eff$, $\yld$ and $\act$ imply
that, for each $s \in S$, $H_s$ is a service indeed.
It is worth mentioning that $\derive{m} H_s = H_{\eff(m,s)}$,
$\rfunc{H_s}(\seq{m}) = \yld(m,s)$, and
$\afunc{H_s}(\seq{m}) = \act(m,s)$.

\section{Services for Molecular Dynamics}
\label{sect-MDS}

In this section, we describe a family of services which concerns
molecular dynamics.
The formal description given here elaborates on an informal description
of molecular dynamics given in~\cite{BB02a}.

The states of molecular dynamics services resemble collections of
molecules composed of atoms and the methods of molecular dynamics
services transform the structure of molecules like in chemical
reactions.
An atom can have \emph{fields} and each of those fields can contain an
atom.
An atom together with the ones it has links to via fields can be viewed
as a submolecule, and a submolecule that is not contained in a larger
submolecule can be viewed as a molecule.
Thus, the collection of molecules that make up a state can be viewed as
a fluid.
By means of methods, new atoms can be created, fields can be added to
and removed from atoms, and the contents of fields of atoms can be
examined and modified.
A few methods use a \emph{spot} to put an atom in or to get an atom
from.
By means of methods, the contents of spots can be compared and modified
as well.
Creating an atom is thought of as turning an element of a given set of
\emph{proto-atoms} into an atom.
If there are no proto-atoms left, then atoms can no longer be created.

It is assumed that a finite set $\Spot$ of \emph{spots} such that
$\Foci,\Meth \subseteq \Spot$, a total order $\leq$ on $\Spot$, and a
finite set $\Field$ of \emph{fields} have been given.
It is further assumed that a countable set $\PAtom$ of
\emph{proto-atoms} such that $\bot \not\in \PAtom$ and a bijection
$\funct{\proatom}{[1,\card(\PAtom)]}{\PAtom}$ have been given.
Although the set of proto-atoms may be infinite, there exists at any
time only a finite number of atoms.
Each of those atoms has only a finite number of fields.
Molecular dynamics services accept the following methods:
\begin{itemize}
\item
for each $s \in \Spot$, a \emph{create atom method} $\creatom{s}$;
\item
for each $s,s' \in \Spot$, a \emph{set spot method} $\setspot{s}{s'}$;
\item
for each $s, \in \Spot$, a \emph{clear spot method} $\clrspot{s}$;
\item
for each $s,s' \in \Spot$,
an \emph{equality test method} $\equaltst{s}{s'}$;
\item
for each $s \in \Spot$,
an \emph{undefinedness test method} $\undeftst{s}$;
\item
for each $s \in \Spot$ and $v \in \Field$,
a \emph{add field method} $\addfield{s}{v}$;
\item
for each $s \in \Spot$ and $v \in \Field$,
a \emph{remove field method} $\rmvfield{s}{v}$;
\item
for each $s \in \Spot$ and $v \in \Field$,
a \emph{has field method} $\hasfield{s}{v}$;
\item
for each $s \in \Spot$ and $v \in \Field$,
a \emph{set field method} $\setfield{s}{v}{s'}$;
\item
for each $s \in \Spot$ and $v \in \Field$,
a \emph{get field method} $\getfield{s}{s'}{v}$;
\item
for each $s,s' \in \Spot$,
a \emph{generate action method} $\genact{s}{s'}$.
\end{itemize}
We write $\Meth_\ga$ for the set of all generate action methods and
$\Meth_\md$ for the set of all methods of molecular dynamics services.
It is assumed that $\Meth_\md \subseteq \Meth$.

The states of molecular dynamics services comprise the contents of all
spots, the fields of the existing atoms, and the contents of those
fields.
The methods accepted by molecular dynamics services can be explained as
follows:
\begin{itemize}
\item
$\creatom{s}$:
if an atom can be created, then the contents of spot $s$ becomes a newly
created atom and the reply is $\True$; otherwise, nothing changes and
the reply is $\False$;
\item
$\setspot{s}{s'}$:
the contents of spot $s$ becomes the same as the contents of spot $s'$
and the reply is $\True$;
\item
$\clrspot{s}$:
the contents of spot $s$ becomes undefined and the reply is $\True$;
\item
$\equaltst{s}{s'}$:
if the contents of spot $s$ equals the contents of spot $s'$, then
nothing changes and the reply is $\True$; otherwise, nothing changes and
the reply is $\False$;
\item
$\undeftst{s}$:
if the contents of spot $s$ is undefined, then nothing changes and the
reply is $\True$; otherwise, nothing changes and the reply is $\False$;
\item
$\addfield{s}{v}$:
if the contents of spot $s$ is an atom and $v$ is not yet a field of
that atom, then $v$ is added (with undefined contents) to the fields of
that atom and the reply is $\True$; otherwise, nothing changes and the
reply is $\False$;
\item
$\rmvfield{s}{v}$:
if the contents of spot $s$ is an atom and $v$ is a field of that atom,
then $v$ is removed from the fields of that atom and the reply is
$\True$; otherwise, nothing changes and the reply is $\False$;
\item
$\hasfield{s}{v}$:
if the contents of spot $s$ is an atom and $v$ is a field of that atom,
then nothing changes and the reply is $\True$; otherwise, nothing
changes and the reply is $\False$;
\item
$\setfield{s}{v}{s'}$:
if the contents of spot $s$ is an atom and $v$ is a field of that atom,
then the contents of that field becomes the same as the contents of spot
$s'$ and the reply is $\True$; otherwise, nothing changes and the reply
is $\False$;
\item
$\getfield{s}{s'}{v}$:
if the contents of spot $s'$ is an atom and $v$ is a field of that atom,
then the contents of spot $s$ becomes the same as the contents of that
field and the reply is $\True$; otherwise, nothing changes and the reply
is $\False$;
\item
$\genact{s}{s'}$:
if the contents of spots $s$ and $s'$ are atoms and there exist
$f' \in \Foci$ and  $m' \in \Meth$ such that the contents of spot $f'$
equals the contents of spot $s$ and the contents of spot $m'$ equals the
contents of spot $s'$, then nothing changes, there is no reply, and
$\genact{s}{s'}$ is turned into $f.m$ where $f$ and $m$ are the least
$f' \in \Foci$ and $m' \in \Meth$ with respect to $\leq$ that satisfy
the conditions formulated above; otherwise, $\genact{s}{s'}$ is
rejected.
\end{itemize}
In the explanation given above, wherever we say that the reply is
$\True$ or the reply is $\False$, this is meant to imply that the method
concerned is turned into $\Tau$.
Moreover, wherever we say that the contents of a spot or field becomes
the same as the contents of another spot or field, this is meant to
imply that the former contents becomes undefined if the latter contents
is undefined.

Let
\begin{ldispl}
\begin{aeqns}
\nm{SS}  & = & \Spot \to (\PAtom \union \set{\bot})\;,
\eqnsep
\nm{AS}  & = &
\Union{A \in \fsetof{(\PAtom)}}
 (A \to
  \Union{F \in \fsetof{(\Field)}} (F \to (\PAtom \union \set{\bot})))\;,
\eqnsep
\nm{MDS} & = &
\set{\tup{\sigma,\alpha} \in \nm{SS} \x \nm{AS} \where
     \rng(\sigma) \subseteq \dom(\alpha) \union \set{\bot} \And {}
\\ & & \hfill
         \Forall{a \in \dom(\alpha)}
          {\rng(\alpha(a)) \subseteq
           \dom(\alpha) \union \set{\bot}}}\;,
\eqnsep
s        & \in & \nm{MDS}\;.
\end{aeqns}
\end{ldispl}
Then we write $\MDS_s$ for the service with initial state $s$ described
by $S = \nm{MDS}\, \union \set{\undef}$, where
$\undef \not\in \nm{MDS}$, and the functions $\eff$, $\yld$ and $\act$
defined in Tables~\ref{eff-mds}, \ref{yld-mds} and~\ref{act-mds}.%
\footnote
{We use the following notation for functions:
 $\emptymap$ for the empty function;
 $\maplet{d}{r}$ for the function $f$ with $\dom(f) = \set{d}$ such that
 $f(d) = r$;
 $f \owr g$ for the function $h$ with $\dom(h) = \dom(f) \union \dom(g)$
 such that for all $d \in \dom(h)$,\, $h(d) = f(d)$ if
 $d \not\in \dom(g)$ and $h(d) = g(d)$ otherwise;
 and $f \dsub D$ for the function $g$ with $\dom(g) = \dom(f) \diff D$
 such that for all $d \in \dom(g)$,\, $g(d) = f(d)$.}%
\begin{table}[!t]
\caption{Effect function for molecular dynamics services}
\label{eff-mds}
\begin{eqntbl}
\begin{axcol}
\eff(\creatom{s},\tup{\sigma,\alpha}) = {} \\ \;\;
\tup{\sigma \owr \maplet{s}{\newatom(\dom(\alpha))},
     \alpha \owr \maplet{\newatom(\dom(\alpha))}{\emptymap}}
 & \mif \newatom(\dom(\alpha)) \neq \bot
\\
\eff(\creatom{s},\tup{\sigma,\alpha}) = \tup{\sigma,\alpha}
 & \mif \newatom(\dom(\alpha)) = \bot
\\
\eff(\setspot{s}{s'},\tup{\sigma,\alpha}) =
\tup{\sigma \owr \maplet{s}{\sigma(s')},\alpha}
\\
\eff(\clrspot{s},\tup{\sigma,\alpha}) =
\tup{\sigma \owr \maplet{s}{\bot},\alpha}
\\
\eff(\equaltst{s}{s'},\tup{\sigma,\alpha}) = \tup{\sigma,\alpha}
\\
\eff(\undeftst{s},\tup{\sigma,\alpha}) = \tup{\sigma,\alpha}
\\
\eff(\addfield{s}{v},\tup{\sigma,\alpha}) = {} \\ \;\;
\tup{\sigma,
     \alpha \owr
     \maplet{\sigma(s)}{\alpha(\sigma(s)) \owr \maplet{v}{\bot}}}
 & \mif \sigma(s) \neq \bot \And v \not\in \dom(\alpha(\sigma(s)))
\\
\eff(\addfield{s}{v},\tup{\sigma,\alpha}) = \tup{\sigma,\alpha}
 & \mif \sigma(s) = \bot \Or v \in \dom(\alpha(\sigma(s)))
\\
\eff(\rmvfield{s}{v},\tup{\sigma,\alpha}) =
\tup{\sigma,
     \alpha \owr
     \maplet{\sigma(s)}{\alpha(\sigma(s)) \dsub \set{v}}}
 & \mif \sigma(s) \neq \bot \And v \in \dom(\alpha(\sigma(s)))
\\
\eff(\rmvfield{s}{v},\tup{\sigma,\alpha}) = \tup{\sigma,\alpha}
 & \mif \sigma(s) = \bot \Or v \not\in \dom(\alpha(\sigma(s)))
\\
\eff(\hasfield{s}{v},\tup{\sigma,\alpha}) = \tup{\sigma,\alpha}
\\
\eff(\setfield{s}{v}{s'},\tup{\sigma,\alpha}) = {} \\ \;\;
\tup{\sigma,
     \alpha \owr
     \maplet
      {\sigma(s)}
      {\alpha(\sigma(s)) \owr \maplet{v}{\sigma(s')}}}
 & \mif \sigma(s) \neq \bot \And v \in \dom(\alpha(\sigma(s)))
\\
\eff(\setfield{s}{v}{s'},\tup{\sigma,\alpha}) = \tup{\sigma,\alpha}
 & \mif \sigma(s) = \bot \Or v \not\in \dom(\alpha(\sigma(s)))
\\
\eff(\getfield{s}{s'}{v},\tup{\sigma,\alpha}) =
\tup{\sigma \owr \maplet{s}{\alpha(\sigma(s'))(v)},\alpha}
 & \mif \sigma(s') \neq \bot \And v \in \dom(\alpha(\sigma(s')))
\\
\eff(\getfield{s}{s'}{v},\tup{\sigma,\alpha}) = \tup{\sigma,\alpha}
 & \mif \sigma(s') = \bot \Or v \not\in \dom(\alpha(\sigma(s')))
\\
\eff(\genact{s}{s'},\tup{\sigma,\alpha}) = \tup{\sigma,\alpha}
 & \mif \gacnd(\sigma,s,s') = \True
\\
\eff(\genact{s}{s'},\tup{\sigma,\alpha}) = \undef
 & \mif \gacnd(\sigma,s,s') = \False
\\
\eff(m,\tup{\sigma,\alpha}) = \undef
 & \mif m \not\in \Meth_\md
\\
\eff(m,\undef) = \undef
\end{axcol}
\end{eqntbl}
\end{table}
\begin{table}[!t]
\caption{Yield function for molecular dynamics services}
\label{yld-mds}
\begin{eqntbl}
\begin{axcol}
\yld(\creatom{s},\tup{\sigma,\alpha}) = \True
 & \mif \newatom(\dom(\alpha)) \neq \bot
\\
\yld(\creatom{s},\tup{\sigma,\alpha}) = \False
 & \mif \newatom(\dom(\alpha)) = \bot
\\
\yld(\setspot{s}{s'},\tup{\sigma,\alpha}) = \True
\\
\yld(\clrspot{s},\tup{\sigma,\alpha}) = \True
\\
\yld(\equaltst{s}{s'},\tup{\sigma,\alpha}) = \True
 & \mif \sigma(s) = \sigma(s')
\\
\yld(\equaltst{s}{s'},\tup{\sigma,\alpha}) = \False
 & \mif \sigma(s) \neq \sigma(s')
\\
\yld(\undeftst{s},\tup{\sigma,\alpha}) = \True
 & \mif \sigma(s) = \bot
\\
\yld(\undeftst{s},\tup{\sigma,\alpha}) = \False
 & \mif \sigma(s) \neq \bot
\\
\yld(\addfield{s}{v},\tup{\sigma,\alpha}) = \True
 & \mif \sigma(s) \neq \bot \And v \not\in \dom(\alpha(\sigma(s)))
\\
\yld(\addfield{s}{v},\tup{\sigma,\alpha}) = \False
 & \mif \sigma(s) = \bot \Or v \in \dom(\alpha(\sigma(s)))
\\
\yld(\rmvfield{s}{v},\tup{\sigma,\alpha}) = \True
 & \mif \sigma(s) \neq \bot \And v \in \dom(\alpha(\sigma(s)))
\\
\yld(\rmvfield{s}{v},\tup{\sigma,\alpha}) = \False
 & \mif \sigma(s) = \bot \Or v \not\in \dom(\alpha(\sigma(s)))
\\
\yld(\hasfield{s}{v},\tup{\sigma,\alpha}) = \True
 & \mif \sigma(s) \neq \bot \And v \in \dom(\alpha(\sigma(s)))
\\
\yld(\hasfield{s}{v},\tup{\sigma,\alpha}) = \False
 & \mif \sigma(s) = \bot \Or v \not\in \dom(\alpha(\sigma(s)))
\\
\yld(\setfield{s}{v}{s'},\tup{\sigma,\alpha}) = \True
 & \mif \sigma(s) \neq \bot \And v \in \dom(\alpha(\sigma(s)))
\\
\yld(\setfield{s}{v}{s'},\tup{\sigma,\alpha}) = \False
 & \mif \sigma(s) = \bot \Or v \not\in \dom(\alpha(\sigma(s)))
\\
\yld(\getfield{s}{s'}{v},\tup{\sigma,\alpha}) = \True
 & \mif \sigma(s') \neq \bot \And v \in \dom(\alpha(\sigma(s')))
\\
\yld(\getfield{s}{s'}{v},\tup{\sigma,\alpha}) = \False
 & \mif \sigma(s') = \bot \Or v \not\in \dom(\alpha(\sigma(s')))
\\
\yld(\genact{s}{s'},\tup{\sigma,\alpha}) = \Mless
 & \mif \gacnd(\sigma,s,s') = \True
\\
\yld(\genact{s}{s'},\tup{\sigma,\alpha}) = \Blocked
 & \mif \gacnd(\sigma,s,s') = \False
\\
\yld(m,\tup{\sigma,\alpha}) = \Blocked
 & \mif m \not\in \Meth_\md
\\
\yld(m,\undef) = \Blocked
\end{axcol}
\end{eqntbl}
\end{table}
\begin{table}[!t]
\caption{Action function for molecular dynamics services}
\label{act-mds}
\begin{eqntbl}
\begin{axcol}
\act(m,\tup{\sigma,\alpha}) = \Tau
 & \mif m \not\in \Meth_\ga
\\
\act(\genact{s}{s'},\tup{\sigma,\alpha}) = \gares(\sigma,s,s')
 & \mif \gacnd(\sigma,s,s') = \True
\\
\act(\genact{s}{s'},\tup{\sigma,\alpha}) = \Tau
 & \mif \gacnd(\sigma,s,s') = \False
\\
\act(m,\undef) = \Tau
\end{axcol}
\end{eqntbl}
\end{table}%
In these tables, the functions
$\funct{\newatom}{\fsetof{(\PAtom)}}{(\PAtom \union \set{\bot})}$,
$\funct{\gacnd}{\nm{SS} \x \Spot \x \Spot}{\set{\True,\False}}$, and
$\funct{\gares}{\nm{SS} \x \Spot \x \Spot}{\BAct}$ are used.
These functions are defined as follows:
\begin{ldispl}
\begin{aceqns}
\newatom(A) & = & \proatom(m + 1) & \mif m  <   \card(\PAtom)\;,
\\
\newatom(A) & = & \bot            & \mif m \geq \card(\PAtom)\;,
\end{aceqns}
\end{ldispl}
where $m = \max \set{n \where \proatom(n) \in A}$;
\begin{ldispl}
\gacnd(\sigma,s,s') = \True\; \mathsf{iff}
\\ \;\; \sigma(s) \neq \bot \And \sigma(s') \neq \bot \And
        \Exists{f \in \Foci}{\sigma(f) = \sigma(s)} \And
        \Exists{m \in \Meth}{\sigma(m) = \sigma(s')}\;;
\end{ldispl}
\begin{ldispl}
\gares(\sigma,s,s') = f.m\;,
\end{ldispl}
where $f = \min \set{f' \in \Foci \where \sigma(f') = \sigma(s)}$
and   $m = \min \set{m' \in \Meth \where \sigma(m') = \sigma(s)}$.
We write $\MDS$ for the family of services described by $S$, $\eff$,
$\yld$, and $\act$.
We write $\MDS_\init$ for the service from this family of which the
initial state is the unique $\tup{\sigma,\alpha} \in S$ such that
$\dom(\alpha) = \emptymap$.

Let $\tup{\sigma,\alpha} \in S$, let $s \in \Spot$, let
$a \in \dom(\alpha)$, and let $v \in \dom(\alpha(a))$.
Then $\sigma(s)$ is the contents of spot $s$ if $\sigma(s) \neq \bot$,
$v$ is a field of atom $a$, and $\alpha(a)(v)$ is the contents of field
$v$ of atom $a$ if $\alpha(a)(v) \neq \bot$.
The contents of spot $s$ is undefined if $\sigma(s) = \bot$, and the
contents of field $v$ of atom $a$ is undefined if $\alpha(a)(v) = \bot$.
Notice that $\dom(\alpha)$ is taken for the set of all existing atoms.
Therefore, the contents of each spot, i.e.\ each element of
$\rng(\sigma)$, must be in $\dom(\alpha)$ if the contents is defined.
Moreover, for each existing atom $a$, the contents of each of its
fields, i.e.\ each element of $\rng(\alpha(a))$, must be in
$\dom(\alpha)$ if the contents is defined.
Molecular dynamics services get into state $\undef$ when refusing a
request to process a command.
The molecular dynamics service with initial state $\undef$ is the same
as the divergent service $\DivServ$.
%
The function $\newatom$ turns proto-atoms into atoms.
After all proto-atoms have been turned into atoms, $\newatom$ yields
$\bot$.
This can only happen if the number of proto-atoms is finite.
The initial state of $\MDS_\init$ is the unique state in which no
proto-atoms have been turned into atoms yet.

\section{Representing Programs by Molecules}
\label{sect-representation}

In this section, we associate each \PGLD\ program with a \PGA\ program
for constructing a representation of the \PGLD\ program by a molecule.

Let $u_1 \conc \ldots \conc u_k$ be a \PGLD\ program in which the foci
$f_1,\ldots,f_n$ and the methods $m_1,\ldots,m_{n'}$ occur.
The idea is that:
\begin{iteml}
\item
an atom is created for each of the foci $f_1,\ldots,f_n$ and the methods
$m_1,\ldots,m_{n'}$, using the focus or method itself as the spot into
which the atom concerned is brought on its creation, and an atom is
created for each of the instructions $u_1,\ldots,u_k$;
\item
each atom that corresponds to an instruction is linked via fields to
other atoms as follows:
\begin{iteml}
\item
if the corresponding instruction is of the form $f.m$, $\ptst{f.m}$ or
$\ntst{f.m}$, then the atom is linked to the atoms that correspond to
the focus and method concerned and the atoms that correspond to the
instructions with which execution must proceed in the cases of a positive
and a negative reply;
\item
if the corresponding instruction is of the form $\ajmp{l}$ and
$1 \leq l \leq k$, then the atom is linked to the atom that corresponds
to the instruction with which execution must proceed, i.e.\ the $l$-th
instruction;
\end{iteml}
\item
to prevent that an exception must be made of the instruction $u_k$ in
the case where it is of the form $f.m$, $\ptst{f.m}$ or $\ntst{f.m}$,
two additional atoms are created that are not linked to other atoms;
\item
the atom that corresponds to the first instruction is brought into spot
$s$.
\end{iteml}
Notice that, if an atom corresponds to an instruction of the form
$\ajmp{l}$ and not $1 \leq l \leq k$, then the atom is not linked to
another atom.

We make the assumptions that $\Foci$ and $\Meth$ are disjoint and that
$\PAtom$ is infinite.
Under these assumptions, atom creation always leads to a positive reply
and no test instructions are needed in the programs for constructing
representations of \PGLD\ programs.
The first assumption is only made because it permits this
simplification.
The second assumption is primarily made because there will be \PGLD\
programs for which there are no \PGA\ programs for constructing their
representation if $\PAtom$ is finite.

We define the \PGA\ programs for constructing representations of \PGLD\
programs by means of a function
$\pgldmd$ from the set of all \PGLD\ programs to the set of all \PGA\
programs.
This function is defined by
\begin{ldispl}
\begin{aeqns}
\pgldmd(u_1 \conc \ldots \conc u_k) & = &
\creatom{f_1} \conc \ldots \conc \creatom{f_n} \conc
\creatom{m_1} \conc \ldots \conc \creatom{m_{n'}} \conc
\creatom{s_1} \conc \ldots \conc \creatom{s_{k+2}} \conc
{} \\ & &
\rho_1(u_1) \conc \ldots \conc \rho_k(u_k) \conc
{} \\ & &
\addfield{s_{k+1}}{\stopf} \conc \addfield{s_{k+2}}{\stopf} \conc
{} \\ & &
\setspot{s}{s_1} \conc
\halt\;,
\end{aeqns}
\end{ldispl}
where the auxiliary functions $\rho_j$ from the set of all primitive
instructions of \PGLD\ to the set of all \PGA\ programs are defined as
follows ($1 \leq j \leq k$):
\begin{ldispl}
\begin{aeqns}
\rho_j(f.m) & = &
\addfield{s_j}{\focusf} \conc \addfield{s_j}{\methodf} \conc
\addfield{s_j}{\posf} \conc \addfield{s_j}{\negf} \conc
{} \\ & &
\setfield{s_j}{\focusf}{f} \conc \setfield{s_j}{\methodf}{m} \conc
\setfield{s_j}{\posf}{s_{j+1}} \conc
\setfield{s_j}{\negf}{s_{j+1}}\;, \\
\rho_j(\ptst{f.m}) & = &
\addfield{s_j}{\focusf} \conc \addfield{s_j}{\methodf} \conc
\addfield{s_j}{\posf} \conc \addfield{s_j}{\negf} \conc
 {} \\ & &
\setfield{s_j}{\focusf}{f} \conc \setfield{s_j}{\methodf}{m} \conc
\setfield{s_j}{\posf}{s_{j+1}} \conc
\setfield{s_j}{\negf}{s_{j+2}}\;, \\
\rho_j(\ntst{f.m}) & = &
\addfield{s_j}{\focusf} \conc \addfield{s_j}{\methodf} \conc
\addfield{s_j}{\posf} \conc \addfield{s_j}{\negf} \conc
 {} \\ & &
\setfield{s_j}{\focusf}{f} \conc \setfield{s_j}{\methodf}{m} \conc
\setfield{s_j}{\posf}{s_{j+2}} \conc
\setfield{s_j}{\negf}{s_{j+1}}\;,
\end{aeqns}
\end{ldispl}
\begin{ldispl}
\begin{aceqns}
\rho_j(\ajmp{l}) \hspace*{.5em} & = &
\addfield{s_j}{\ajmpf} \conc \setfield{s_j}{\ajmpf}{s_l} \quad
 & \mif 1 \leq l \leq k\;, \\
\rho_j(\ajmp{l}) & = & \addfield{s_j}{\stopf}
 & \mif \Not (1 \leq l \leq k)
\end{aceqns}
\end{ldispl}
and
\begin{iteml}
\item
$f_1,\ldots,f_n \in \Foci$ are the different foci that occur in
$u_1 \conc \ldots \conc u_k$;
\item
$m_1,\ldots,m_{n'} \in \Meth$ are the different methods that occur in
$u_1 \conc \ldots \conc u_k$;
\item
$s,s_1,\ldots,s_{k+2} \in \Spot \diff (\Foci \union \Meth)$.
\end{iteml}
In this definition, the omitted focus is considered to be $\md$.

The properties stated in the following proposition show that it is easy
to retrieve the \PGLD\ program $P$ from its representation constructed by
$\pgldmd(P)$.
\begin{proposition}
\label{prop-representation}
Let $P = u_1 \conc \ldots \conc u_k$ be a \PGLD\ program, and
let $\tup{\sigma,\alpha}$ be such that
$\MDS_{\tup{\sigma,\alpha}} =
 \apply{\extr{\pgldmd(P)}}{\md}{\MDS_\init}$.
Then $\yld(\equaltst{s}{s_1},\tup{\sigma,\alpha}) = \True$ and
for all $j,l \in [1,k]$, $f \in \Foci$ and $m \in \Meth$:
\begin{iteml}
\normalfont
\item
$u_j = f.m$ iff \\ ${} \quad$
\begin{tabular}[t]{@{}ll@{}}
$\yld(\getfield{u}{s_j}{\focusf},\tup{\sigma,\alpha}) = \True$, &
$\yld(\equaltst{u}{f},
      \eff(\getfield{u}{s_j}{\focusf},\tup{\sigma,\alpha})) = \True$, \\
$\yld(\getfield{v}{s_j}{\methodf},\tup{\sigma,\alpha}) = \True$, &
$\yld(\equaltst{v}{m},
      \eff(\getfield{u}{s_j}{\methodf},\tup{\sigma,\alpha})) = \True$, \\
$\yld(\getfield{s}{s_j}{\posf},\tup{\sigma,\alpha}) = \True$, &
$\yld(\equaltst{s}{s_{j+1}},
      \eff(\getfield{s}{s_j}{\posf},\tup{\sigma,\alpha})) = \True$, \\
$\yld(\getfield{s}{s_j}{\negf},\tup{\sigma,\alpha}) = \True$, &
$\yld(\equaltst{s}{s_{j+1}},
      \eff(\getfield{s}{s_j}{\negf},\tup{\sigma,\alpha})) = \True$;
\end{tabular}
\item
$u_j = \ptst{f.m}$ iff \\ ${} \quad$
\begin{tabular}[t]{@{}ll@{}}
$\yld(\getfield{u}{s_j}{\focusf},\tup{\sigma,\alpha}) = \True$, &
$\yld(\equaltst{u}{f},
      \eff(\getfield{u}{s_j}{\focusf},\tup{\sigma,\alpha})) = \True$, \\
$\yld(\getfield{v}{s_j}{\methodf},\tup{\sigma,\alpha}) = \True$, &
$\yld(\equaltst{v}{m},
      \eff(\getfield{u}{s_j}{\methodf},\tup{\sigma,\alpha})) = \True$, \\
$\yld(\getfield{s}{s_j}{\posf},\tup{\sigma,\alpha}) = \True$, &
$\yld(\equaltst{s}{s_{j+1}},
      \eff(\getfield{s}{s_j}{\posf},\tup{\sigma,\alpha})) = \True$, \\
$\yld(\getfield{s}{s_j}{\negf},\tup{\sigma,\alpha}) = \True$, &
$\yld(\equaltst{s}{s_{j+2}},
      \eff(\getfield{s}{s_j}{\negf},\tup{\sigma,\alpha})) = \True$;
\end{tabular}
\item
$u_j = \ntst{f.m}$ iff \\ ${} \quad$
\begin{tabular}[t]{@{}ll@{}}
$\yld(\getfield{u}{s_j}{\focusf},\tup{\sigma,\alpha}) = \True$, &
$\yld(\equaltst{u}{f},
      \eff(\getfield{u}{s_j}{\focusf},\tup{\sigma,\alpha})) = \True$, \\
$\yld(\getfield{v}{s_j}{\methodf},\tup{\sigma,\alpha}) = \True$, &
$\yld(\equaltst{v}{m},
      \eff(\getfield{u}{s_j}{\methodf},\tup{\sigma,\alpha})) = \True$, \\
$\yld(\getfield{s}{s_j}{\posf},\tup{\sigma,\alpha}) = \True$, &
$\yld(\equaltst{s}{s_{j+2}},
      \eff(\getfield{s}{s_j}{\posf},\tup{\sigma,\alpha})) = \True$, \\
$\yld(\getfield{s}{s_j}{\negf},\tup{\sigma,\alpha}) = \True$, &
$\yld(\equaltst{s}{s_{j+1}},
      \eff(\getfield{s}{s_j}{\negf},\tup{\sigma,\alpha})) = \True$;
\end{tabular}
\item
$u_j = \ajmp{l}$ iff \\ ${} \quad$
$\yld(\getfield{s}{s_j}{\ajmpf},\tup{\sigma,\alpha}) = \True$,
$\yld(\equaltst{s}{s_{l}},
      \eff(\getfield{s}{s_j}{\ajmpf},\tup{\sigma,\alpha})) = \True$;
\item
$u_j = \ajmp{l'}$ for some $l' \not\in [1,k]$ iff
$\yld(\hasfield{s_j}{\stopf},\tup{\sigma,\alpha}) = \True$.
\end{iteml}
\end{proposition}
\begin{proof}
From the assumption that $\Foci$ and $\Meth$ are disjoint, the assumption
that $\PAtom$ is infinite, and the definition of $\pgldmd$, it follows
easily that:
\begin{iteml}
\item
all atom creations are successful;
\item
the content of each spot used in an atom creation does not change after
the atom creation concerned;
\item
$\rho_j(u_j)$ modifies the atom that is the contents of spot $s_j$ only
($1 \leq j \leq k$),
$\addfield{s_{k+1}}{\stopf}$ modifies the atom that is the contents of
spot $s_{k+1}$ only, and
$\addfield{s_{k+2}}{\stopf}$ modifies the atom that is the contents of
spot $s_{k+2}$ only.
\end{iteml}
From this, the properties stated in the proposition follow
straightforwardly by case distinction as in the definition of $\rho_j$.
\qed
\end{proof}

\section{Interpreting Programs Represented by Molecules}
\label{sect-interpretation}

In this section, we introduce a \PGA\ program for interpreting \PGLD\
programs represented by molecules.

The idea is that:
\begin{iteml}
\item
if the atom to be interpreted has the field $\stopf$, then the
interpretation terminates;
\item
if the atom to be interpreted has the field $\ajmpf$, then the
interpretation proceeds with the atom to which it is linked via this
field;
\item
otherwise, first the basic instruction represented by the atoms to which it
is linked via the fields $\focusf$ and $\methodf$ is executed and then
the interpretation proceeds with the atom to which it is linked via the
field $\posf$ or the field $\negf$, depending upon the reply being
positive or negative.
\end{iteml}

The following is the \PGA\ program for interpreting \PGLD\ programs
represented by molecules:
\begin{ldispl}
\begin{aeqns}
(\ptst{\hasfield{s}{\stopf}} \conc \halt \conc
{} \\ \phantom{(}
 \ptst{\hasfield{s}{\ajmpf}} \conc \fjmp{9} \conc
{} \\ \phantom{(}
 \getfield{u}{s}{\focusf} \conc \getfield{v}{s}{\methodf} \conc
 \ptst{\genact{u}{v}} \conc \fjmp{3} \conc
  \getfield{s}{s}{\negf} \conc \fjmp{4} \conc
  \getfield{s}{s}{\posf} \conc \fjmp{2} \conc
{} \\ \phantom{(}
 \getfield{s}{s}{\ajmpf}) \rep\;,
\end{aeqns}
\end{ldispl}
where $u,v \in \Spot \diff (\Foci \union \Meth)$.
Again, the omitted focus is considered to be $\md$.
Below, we write $I$ for this \PGA\ program.

The program $I$ interprets \PGLD\ programs correctly in the sense that,
for all \PGLD\ programs $P$, the intended behaviour of $P$ under
execution coincides with the behaviour of the interpreter under
execution on interaction with the molecular dynamics service that holds
the representation of $P$ constructed by the program $\pgldmd(P)$ when
abstracted from $\Tau$.
This is stated rigorously by Theorem~\ref{theorem-correctness} below.
The theorem is preceded by Proposition~\ref{prop-interpretation} which is
used in the proof of the theorem.
The proposition concerns the local use of the spots $u$ and $v$ in $I$.

For convenience, we introduce some abbreviations.
Let $P = u_1 \conc \ldots \conc u_k$ be a \PGLD\ program and let
$i \in [1,k]$.
Then we write $\MDS_{\repr(P)}$ for
$\apply{\extr{\pgldmd(P)}}{\md}{\MDS_\init}$ and $\MDS_{\repr(i,P)}$ for
$\apply{\extr{\setspot{s}{s_i} \conc \halt}}{\md}{\MDS_{\repr(P)}}$.

\begin{proposition}
\label{prop-interpretation}
Let $P = u_1 \conc \ldots \conc u_k$ be a \PGLD\ program,
let $i \in [1,k]$, and
let $\MDS'_{\repr(i,P)}$ be such that
$\MDS_{\repr(i,P)} =
 \apply{\extr{\clrspot{u} \conc \clrspot{v} \conc \halt}}{\md}
       {\MDS'_{\repr(i,P)}}$.
Then
$\abstr(\use{\extr{I}}{\md}{\MDS_{\repr(i,P)}}) =
 \abstr(\use{\extr{I}}{\md}{\MDS'_{\repr(i,P)}})$.
\end{proposition}
\begin{proof}
By \AIP, it is sufficient to prove that for all $n \in \Nat$, for all
$i \in [1,k]$,
$\proj{n}{\abstr(\use{\extr{I}}{\md}{\MDS_{\repr(i,P)}})} =
 \proj{n}{\abstr(\use{\extr{I}}{\md}{\MDS'_{\repr(i,P)}})}$.
This property is straightforward to prove by induction on $n$.
\qed
\end{proof}

\begin{theorem}
\label{theorem-correctness}
For all \PGLD\ programs $P$:
\begin{ldispl}
\extr{P}_\sPGLD =
\abstr(\use{\extr{I}}{\md}{\MDS_{\repr(P)}})\;.
\end{ldispl}
\end{theorem}
\begin{proof}
In the proof, we make use of an auxiliary function $\extr{\ph,\ph}$ from
the set of all natural numbers and the set of all \PGLD\ programs to the
set of all closed terms of sort $\Thr$.
It gives, for each natural number $i$ and \PGLD\ program
$u_1 \conc \ldots \conc u_k$, a closed term of sort $\Thr$ that denotes
the behaviour of $u_1 \conc \ldots \conc u_k$ when executed from
position $i$ if $1 \leq i \leq k$ and $\Stop$ otherwise.
This function is defined as follows:
\begin{ldispl}
\begin{aceqns}
\extr{i,u_1 \conc \ldots \conc u_k} & = &
\multicolumn{2}{@{}l@{}}
{\extr{\psi_i(u_i) \conc \ldots \conc \psi_k(u_k) \conc
       \halt \conc \halt \conc
       (\psi_1(u_1) \conc \ldots \conc \psi_k(u_k) \conc
        \halt \conc \halt)\rep}
} \\ & &
& \hspace*{17.25em} \mif 1 \leq i \leq k\;,
\\
\extr{i,u_1 \conc \ldots \conc u_k} & =  & \Stop
& \hspace*{17.25em} \mif \Not 1 \leq i \leq k\;
\end{aceqns}
\end{ldispl}
(where $\psi_j$ is as in the definition of $\pgldpga$).
It follows easily from the definition of $\extr{\ph,\ph}$ and the axioms
of \PGA\ that if $1 \leq i \leq k$:
\begin{ldispl}
\begin{aceqns}
\extr{i,u_1 \conc \ldots \conc u_k} & = &
a \bapf \extr{i+1,u_1 \conc \ldots \conc u_k}
& \mif u_i = a\;, \\
\extr{i,u_1 \conc \ldots \conc u_k} & = &
\pcc{\extr{i+1,u_1 \conc \ldots \conc u_k}}{a}
    {\extr{i+2,u_1 \conc \ldots \conc u_k}}
& \mif u_i = \ptst{a}\;, \\
\extr{i,u_1 \conc \ldots \conc u_k} & = &
\pcc{\extr{i+2,u_1 \conc \ldots \conc u_k}}{a}
    {\extr{i+1,u_1 \conc \ldots \conc u_k}}
& \mif u_i = \ntst{a}\;, \\
\extr{i,u_1 \conc \ldots \conc u_k} & = &
\extr{l,u_1 \conc \ldots \conc u_k}
& \mif u_i = \ajmp{l}\;.
\end{aceqns}
\end{ldispl}

Let $P = u_1 \conc \ldots \conc u_k$ be a \PGLD\ program,
let
\begin{ldispl}
T  = \set{\extr{i,P} \where i \in [1,k]}\;, \\
T' = \set{\abstr(\use{\extr{I}}{\md}{\MDS_{\repr(i,P)}}) \where
          i \in [1,k]}\;,
\end{ldispl}
and let $\funct{\beta}{T}{T'}$ be the bijection defined by
\begin{ldispl}
\beta(\extr{i,P}) =
\abstr(\use{\extr{I}}{\md}{\MDS_{\repr(i,P)}})\;.
\end{ldispl}
For each closed term $p'$ of sort $\Thr$, write $\beta^*(p')$ for $p'$
with, for all $p \in T$, all occurrences of $p$ in $p'$ replaced by
$\beta(p)$.
Then, using the equations concerning the auxiliary function
$\extr{\ph,\ph}$ given above and Propositions~\ref{prop-representation}
and~\ref{prop-interpretation}, it is straightforward to prove that there
exists a set~$E$ consisting of one derivable equation $p = p'$ for each
$p \in T$ such that, for all equations $p = p'$ in $E$:
\begin{iteml}
\item
the equation $\beta(p) = \beta^*(p')$ is derivable;
\item
$p' \in T$ only if $p'$ can be rewritten to a $p'' \not\in T$ using the
equations in $E$ from left to right.
\end{iteml}
Because
$\extr{P}_\sPGLD = \extr{1,P}$ and
$\abstr(\use{\extr{I}}{\md}{\MDS_{\repr(P)}}) =
 \abstr(\use{\extr{I}}{\md}{\MDS_{\repr(1,P)}})$,
this means that
$\extr{P}_\sPGLD$ and $\abstr(\use{\extr{I}}{\md}{\MDS_{\repr(P)}})$
are solutions of the same guarded recursive specification.
Because guarded recursive specifications have unique solutions, it
follows immediately that
$\extr{P}_\sPGLD = \abstr(\use{\extr{I}}{\md}{\MDS_{\repr(P)}})$.
\qed
\end{proof}
In the program $\pgldmd(u_1 \conc \ldots \conc u_k)$, we could have
replaced $\setspot{s}{s_1} \conc \halt$ by
$\setspot{s}{s_1} \conc
 \clrspot{s_1} \conc \ldots \conc \clrspot{s_{k+2}} \conc \halt$
because the program $I$ does not use the spots $s_1,\ldots,s_{k+2}$.
However, we cannot conceive a proof of Theorem~\ref{theorem-correctness}
in which these auxiliary spots are not found.
The main problem is that we cannot find a way of formulating the gist of
Proposition~\ref{prop-representation}, which looks to be material to a
proof of the theorem, without referring to the spots
$s_1,\ldots,s_{k+2}$.

\section{\PGLD\ with Indirect Jumps}
\label{sect-PGLDij}

In this section, we introduce a variant of \PGLD\ with indirect jump
instructions.
This variant is called \PGLDij.
However, preceding the introduction of \PGLDij, we give a state-based
description of the very simple family of services that constitute a
register file of which the registers can contain natural numbers up to
some bound.
This register file will be used later on to describe the behaviour of
\PGLDij\ programs.

It is assumed that a fixed but arbitrary number $\maxr$ has been given,
which is considered the number of registers available.
It is also assumed that a fixed but arbitrary number $\maxn$ has been
given, which is considered the greatest natural number that can be
contained in a register.

The register file services accept the following methods:
\begin{itemize}
\item
for each $i \in [0,\maxr]$ and $n \in [0,\maxn]$,
a \emph{register set method} $\setr{:}i{:}n$;
\item
for each $i \in [0,\maxr]$ and $n \in [0,\maxn]$,
a \emph{register test method} $\eqr{:}i{:}n$.
\end{itemize}
We write $\Meth_\rf$ for the set
$\set{\setr{:}i{:}n,\eqr{:}i{:}n \where
      i \in [0,\maxr] \And n \in [0,\maxn]}$.
It is assumed that $\Meth_\rf \subseteq \Meth$.

The methods accepted by register file services can be explained as
follows:
\begin{itemize}
\item
$\setr{:}i{:}n$\,:
the contents of register $i$ becomes $n$ and the reply is $\True$;
\item
$\eqr{:}i{:}n$\,:
if the contents of register $i$ equals $n$, then nothing changes and the
reply is $\True$; otherwise nothing changes and the reply is $\False$.
\end{itemize}

Let $\funct{s}{[1,\maxr]}{[0,\maxn]}$.
Then we write $\RF_s$ for the service with initial state $s$ described
by $S = (\mapof{[1,\maxr]}{[0,\maxn]}) \union \set{\undef}$, where
$\undef \not\in \mapof{[1,\maxr]}{[0,\maxn]}$, and the functions $\eff$,
$\yld$ and $\act$ defined as follows ($i \in [0,\maxr]$,
$n \in [0,\maxn]$, $\funct{\rho}{[1,\maxr]}{[0,\maxn]}$):
\begin{ldispl}
\begin{gceqns}
\eff(\setr{:}i{:}n,\rho) = \rho \owr \maplet{i}{n}\;,
\\
\eff(\eqr{:}i{:}n,\rho)  = \rho\;,
\\
\eff(m,\rho)      = \undef       & \mif m \not\in \Meth_\rf\;,
\\
\eff(m,\undef) = \undef\;,
\eqnsep
\yld(\setr{:}i{:}n,\rho) = \True\;,
\\
\yld(\eqr{:}i{:}n,\rho) = \True  & \mif \rho(i) = n\;,
\\
\yld(\eqr{:}i{:}n,\rho) = \False & \mif \rho(i) \neq n\;,
\\
\yld(m,\rho)      = \Blocked     & \mif m \not\in \Meth_\rf\;,
\\
\yld(m,\undef) = \Blocked\;,
\eqnsep
\act(m,\rho) = \Tau\;,
\\
\act(m,\undef) = \Tau\;.
\end{gceqns}
\end{ldispl}%
We write $\RF_\mathrm{init}$ for
$\RF_{\maplet{1}{0} \owr \ldots \owr \maplet{I}{0}}$.

We pass on to \PGLD\ with indirect jump instructions.
In \PGLDij, it is assumed that there is a fixed but arbitrary finite set
of \emph{foci} $\Foci$ with $\rf \in \Foci$ and a fixed but arbitrary
finite set of \emph{methods} $\Meth$.
Moreover, we adopt the assumptions made about register file services
above.
The set $\set{f.m \where f \in \Foci \diff \set{\rf}, m \in \Meth}$ is
taken as the set $\BInstr$ of basic instructions.

\PGLDij\ has the following primitive instructions:
\begin{iteml}
\item
for each $a \in \BInstr$, a \emph{plain basic instruction} $a$;
\item
for each $a \in \BInstr$, a \emph{positive test instruction} $\ptst{a}$;
\item
for each $a \in \BInstr$, a \emph{negative test instruction} $\ntst{a}$;
\item
for each $l \in \Nat$, a \emph{direct absolute jump instruction}
$\ajmp{l}$;
\item
for each $i \in [1,\maxr]$ and $n \in [1,\maxn]$,
a \emph{register set instruction} $\setr{:}i{:}n$;
\item
for each $i \in [1,\maxr]$, an \emph{indirect absolute jump instruction}
$\iajmp{i}$.
\end{iteml}
\PGLDij\ programs have the form $u_1 \conc \ldots \conc u_k$, where
$u_1,\ldots,u_k$ are primitive instructions of \PGLDij.

The plain basic instructions, the positive test instructions, the
negative test instructions, and the direct absolute jump instructions
are as in \PGLD.
The effect of a register set instruction $\setr{:}i{:}n$ is that the
content of register $i$ becomes~$n$.
The effect of an indirect absolute jump instruction $\iajmp{i}$ is that
execution proceeds with the $l$-th instruction of the program concerned,
where $l$ is the content of register $i$.
If $\iajmp{i}$ is itself the $l$-th instruction, then deadlock occurs.
If $l$ equals $0$ or $l$ is greater than the length of the program,
termination occurs.
We stipulate that the content of each register is initially $0$.

Like before, we define the meaning of \PGLDij\ programs by means of a
function $\pgldijpgld$ from the set of all \PGLDij\ programs to the set
of all \PGLD\ programs.
This function is defined by
\begin{ldispl}
\pgldijpgld(u_1 \conc \ldots \conc u_k) = \\ \quad
\psi(u_1) \conc \ldots \conc \psi(u_k) \conc
\ajmp{0} \conc \ajmp{0}  \conc {} \\ \quad
\ptst{\rf.\eqr{:}1{:}1} \conc \ajmp{1} \conc \ldots \conc
\ptst{\rf.\eqr{:}1{:}n} \conc \ajmp{n} \conc \ajmp{0} \conc {} \\
\qquad \vdots  \\ \quad
\ptst{\rf.\eqr{:}\maxr{:}1} \conc \ajmp{1} \conc \ldots \conc
\ptst{\rf.\eqr{:}\maxr{:}n} \conc \ajmp{n} \conc \ajmp{0}\;,
\end{ldispl}%
where $n = \min(k,\maxn)$ and the auxiliary function $\psi$ from the set
of all primitive instructions of \PGLDij\ to the set of all primitive
instructions of \PGLD\ is defined as follows:
\begin{ldispl}
\begin{aceqns}
\psi(a)         & = & a\;, \\
\psi(\ptst{a})  & = & \ptst{a}\;, \\
\psi(\ntst{a})  & = & \ntst{a}\;, \\
\psi(\ajmp{l})  & = & \ajmp{l} & \mif l \leq k\;, \\
\psi(\ajmp{l})  & = & \ajmp{0} & \mif l   >  k\;, \\
\psi(\setr{:}i{:}n) & = & \rf.\setr{:}i{:}n\;, \\
\psi(\iajmp{i}) & = & \ajmp{l_i}\;,
\end{aceqns}
\end{ldispl}%
and for each $i \in [1,\maxr]$:
\begin{ldispl}
\begin{aeqns}
l_i & = & k + 3 + (2 \mul \min(k,\maxn) + 1) \mul (i - 1)\;.
\end{aeqns}
\end{ldispl}%
The idea is that each indirect absolute jump can be replaced by a direct
absolute jump to the beginning of the instruction sequence
\begin{ldispl}
\begin{aeqns}
\ptst{\rf.\eqr{:}i{:}1} \conc \ajmp{1} \conc \ldots \conc
\ptst{\rf.\eqr{:}i{:}n} \conc \ajmp{n} \conc \ajmp{0}\;,
\end{aeqns}
\end{ldispl}%
where $i$ is the register concerned and $n = \min(k,\maxn)$.
The execution of this instruction sequence leads to the intended jump
after the content of the register concerned has been found by a linear
search.
To enforce termination of the program after execution of its last
instruction if the last instruction is a plain basic instruction, a
positive test instruction or a negative test instruction,
$\ajmp{0} \conc \ajmp{0}$ is appended to
$\psi(u_1) \conc \ldots \conc \psi(u_k)$.
Because the length of the translated program is greater than $k$, care
is taken that there are no direct absolute jumps to instructions with a
position greater than $k$.
Obviously, the linear search for the content of a register can be
replaced by a binary search.

Let $P$ be a \PGLDij\ program.
Then $\pgldijpgld(P)$ represents the meaning of $P$ as a \PGLD\ program.
The intended behaviour of $P$ is the behaviour of $\pgldijpgld(P)$ on
interaction with a register file when abstracted from $\Tau$.
That is, the \emph{behaviour} of $P$, written $\extr{P}_\sPGLDij$, is
$\abstr(\use{\extr{\pgldijpgld(P)}_\sPGLD}{\rf}{\RF_\mathrm{init}})$.


A slightly different variant of \PGLD\ with indirect jump instructions
is introduced in~\cite{BM07e} under the same name.

\section{The Interpretation of \PGLDij\ Programs}
\label{sect-interpretation-ij}

In this section, we associate each \PGLDij\ program with a \PGA\ program
for constructing a representation of the \PGLDij\ program by a molecule
and introduce a \PGA\ program for interpreting \PGLDij\ programs
represented by molecules.
This amounts to enhancing the \PGA\ programs given in
Sections~\ref{sect-representation} and~\ref{sect-interpretation}.
\PGLDij\ programs without indirect absolute jump instructions are
represented and interpreted exactly as before.

The idea of the enhancements is that:
\begin{iteml}
\item
an atom is created for each of the registers;
\item
each atom that corresponds to a register is handled as if it concerns a
direct jump instruction to the instruction with the content of the
register as position;
\item
each atom that corresponds to an indirect jump instruction is handled as
if it concerns a direct jump instruction to the direct jump instruction
that takes the place of the register concerned;
\item
each register set instruction is handled as if it concerns an
instruction for overwriting the direct jump instruction that takes the
place of the register concerned.
\end{iteml}

We define the \PGA\ programs for constructing representations of
\PGLDij\ programs by means of a function $\pgldijmd$ from the set of all
\PGLDij\ programs to the set of all \PGA\ programs.
This function is defined by
\pagebreak[2]
\begin{ldispl}
\begin{geqns}
\pgldijmd(u_1 \conc \ldots \conc u_k) =
{} \\ \quad
\creatom{f_1} \conc \ldots \conc \creatom{f_n} \conc
\creatom{m_1} \conc \ldots \conc \creatom{m_{n'}} \conc
\creatom{s_1} \conc \ldots \conc \creatom{s_{k+2}} \conc
\creatom{s'_1} \conc \ldots \conc \creatom{s'_I} \conc
{} \\ \quad
\rho_1(u_1) \conc \ldots \conc \rho_k(u_k) \conc
{} \\ \quad
\addfield{s_{k+1}}{\stopf} \conc \addfield{s_{k+2}}{\stopf} \conc
{} \\ \quad
\addfield{s'_1}{\ajmpf} \conc \setfield{s'_1}{\ajmpf}{s_{k+2}} \conc
 \ldots \conc
\addfield{s'_I}{\ajmpf} \conc \setfield{s'_I}{\ajmpf}{s_{k+2}} \conc
{} \\ \quad
\setspot{s}{s_1} \conc
\halt\;,
\end{geqns}
\end{ldispl}
where the auxiliary functions $\rho_j$ from the set of all primitive
instructions of \PGLDij\ to the set of all \PGA\ programs are defined as
follows ($1 \leq j \leq k$):
\begin{ldispl}
\begin{aeqns}
\rho_j(f.m) & = &
\addfield{s_j}{\focusf} \conc \addfield{s_j}{\methodf} \conc
\addfield{s_j}{\posf} \conc \addfield{s_j}{\negf} \conc
{} \\ & &
\setfield{s_j}{\focusf}{f} \conc \setfield{s_j}{\methodf}{m} \conc
\setfield{s_j}{\posf}{s_{j+1}} \conc
\setfield{s_j}{\negf}{s_{j+1}}\;, \\
\rho_j(\ptst{f.m}) & = &
\addfield{s_j}{\focusf} \conc \addfield{s_j}{\methodf} \conc
\addfield{s_j}{\posf} \conc \addfield{s_j}{\negf} \conc
 {} \\ & &
\setfield{s_j}{\focusf}{f} \conc \setfield{s_j}{\methodf}{m} \conc
\setfield{s_j}{\posf}{s_{j+1}} \conc
\setfield{s_j}{\negf}{s_{j+2}}\;, \\
\rho_j(\ntst{f.m}) & = &
\addfield{s_j}{\focusf} \conc \addfield{s_j}{\methodf} \conc
\addfield{s_j}{\posf} \conc \addfield{s_j}{\negf} \conc
 {} \\ & &
\setfield{s_j}{\focusf}{f} \conc \setfield{s_j}{\methodf}{m} \conc
\setfield{s_j}{\posf}{s_{j+2}} \conc
\setfield{s_j}{\negf}{s_{j+1}}\;,
\end{aeqns}
\\
\begin{aceqns}
\rho_j(\ajmp{l}) \hspace*{.5em} & = &
\addfield{s_j}{\ajmpf} \conc \setfield{s_j}{\ajmpf}{s_l}
 & \mif 1 \leq l \leq k\;, \\
\rho_j(\ajmp{l}) & = & \addfield{s_j}{\stopf}
 & \mif \Not (1 \leq l \leq k)\;, \\
\rho_j(\setr{:}i{:}l) & = &
\addfield{s_j}{\regf} \conc \addfield{s_j}{\contf} \conc
\addfield{s_j}{\nxtf} \conc
 {} \\ & &
\setfield{s_j}{\regf}{s'_i} \conc \setfield{s_j}{\contf}{s_l} \conc
\setfield{s_j}{\nxtf}{s_{j+1}}
 & \mif 1 \leq l \leq k\;, \\
\rho_j(\setr{:}i{:}l) & = &
\addfield{s_j}{\regf} \conc \addfield{s_j}{\contf} \conc
\addfield{s_j}{\nxtf} \conc
 {} \\ & &
\setfield{s_j}{\regf}{s'_i} \conc \setfield{s_j}{\contf}{s_{k+2}} \conc
\setfield{s_j}{\nxtf}{s_{j+1}}
 & \mif \Not (1 \leq l \leq k)\;, \\
\rho_j(\iajmp{i}) & = &
\addfield{s_j}{\ajmpf} \conc \setfield{s_j}{\ajmpf}{s'_i}
\end{aceqns}
\end{ldispl}
and
\begin{iteml}
\item
$f_1,\ldots,f_n \in \Foci$ are the different foci that occur in
$u_1 \conc \ldots \conc u_k$;
\item
$m_1,\ldots,m_{n'} \in \Meth$ are the different methods that occur in
$u_1 \conc \ldots \conc u_k$;
\item
$s,s_1,\ldots,s_{k+2},s'_1,\ldots,s'_I \in
 \Spot \diff (\Foci \union \Meth)$.
\end{iteml}

The following is the \PGA\ program for interpreting \PGLDij\ programs
represented by molecules:
\begin{ldispl}
\begin{aeqns}
(\ptst{\hasfield{s}{\stopf}} \conc \halt \conc
{} \\ \phantom{(}
 \ptst{\hasfield{s}{\ajmpf}} \conc \fjmp{16} \conc
{} \\ \phantom{(}
 \ptst{\hasfield{s}{\regf}} \conc \fjmp{9} \conc
{} \\ \phantom{(}
 \getfield{u}{s}{\focusf} \conc \getfield{v}{s}{\methodf} \conc
 \ptst{\genact{u}{v}} \conc \fjmp{3} \conc
  \getfield{s}{s}{\negf} \conc \fjmp{9} \conc
  \getfield{s}{s}{\posf} \conc \fjmp{7} \conc
{} \\ \phantom{(}
 \getfield{u}{s}{\regf} \conc \getfield{v}{s}{\contf} \conc
 \setfield{u}{\ajmpf}{v} \conc
 \getfield{s}{s}{\nxtf} \conc \fjmp{2} \conc
{} \\ \phantom{(}
 \getfield{s}{s}{\ajmpf}) \rep\;,
\end{aeqns}
\end{ldispl}
where $u,v \in \Spot \diff (\Foci \union \Meth)$.
Below, we write $I'$ for this \PGA\ program.

Theorem~\ref{theorem-correctness-ij} below states rigorously that
program $I'$ interprets \PGLDij\ programs correctly.
In that theorem, other than in Theorem~\ref{theorem-correctness}, we
write $\MDS_{\repr(P)}$ for
$\apply{\extr{\pgldijmd(P)}}{\md}{\MDS_\init}$.
\begin{theorem}
\label{theorem-correctness-ij}
For all \PGLDij\ programs $P$:
\begin{ldispl}
\extr{P}_\sPGLDij =
\abstr(\use{\extr{I'}}{\md}{\MDS_{\repr(P)}})\;.
\end{ldispl}
\end{theorem}
\begin{proof}
The proof follows the same line as the proof of
Theorem~\ref{theorem-correctness}.
This is possible because on interpretation any change of the state of
the register file is reflected by a corresponding change of its
molecular representation.
\qed
\end{proof}

\section{Conclusions}
\label{sect-concl}

In this paper, we have considered the programming of an interpreter for
a program notation that is close to existing assembly languages using
\PGA\ with the primitives of molecular dynamics as basic instructions.
We have given \PGA\ programs for constructing representations of the
programs to be interpreted by molecules and a \PGA\ program for
interpreting those representations and we have shown that the latter
\PGA\ program does the interpretation correctly.
We have experienced that, although primarily meant for explaining
programming language features relating to the use of dynamic data
structures, the collection of primitives of molecular dynamics in itself
is suited to the programming wants concerned.

We observe that:
(i)~the program notation for which the presented interpreter has been
designed belongs to the simplest program notations devised ever,
(ii)~it is hard to imagine that the programs to be interpreted can be
represented by molecules in a way that is simpler than the way chosen
for the presented interpreter,
(iii)~it is hard to conceive of an interpreter that is simpler than the
presented interpreter.
This means not at all that the design of the interpreter was simple.
On the contrary, the design turned out to be disappointingly difficult.
It happened that, owing to the quest for a simple interpreter, it was
inescapable that the design was to a great extent a trial-and-error
matter.

Dynamic data structures modelled using molecular dynamics can
straightforwardly be implemented in programming languages ranging from
PASCAL~\cite{Wir71a} to C\#~\cite{HWG03a} through pointers or
references, provided that fields are not added or removed dynamically.
Using molecular dynamics, we need not be aware of the existence of the
pointers used for linking data.
The name molecular dynamics refers to the molecule metaphor used in the
introduction.
By that, there is no clue in the name itself to what it stands for.
To remedy this defect, we suggest data linkage dynamics as an
alternative name.

\bibliographystyle{plain}
\bibliography{TA}


\end{document}